\documentclass[10pt,journal,compsoc]{IEEEtran}
\ifCLASSOPTIONcompsoc
  \usepackage[nocompress]{cite}
\else
  \usepackage{cite}
\fi

\normalsize
\usepackage{epsfig, amssymb, amsmath, amsthm, algorithm, algorithmic, color}
\usepackage{cite, subfig, diagbox}

\usepackage{mathtools}
\DeclarePairedDelimiter{\ceil}{\lceil}{\rceil}

\newtheorem{lemma}{Lemma}
\newtheorem{theorem}{Theorem}

\newcommand{\tabincell}[2]{\begin{tabular}{@{}#1@{}}#2\end{tabular}}

\def \be {\begin{eqnarray}}
\def \ee {\end{eqnarray}}
\def \ben {\begin{eqnarray*}}
\def \een {\end{eqnarray*}}
\def \xbs {\begin{split}}
\def \xes {\end{split}}
\begin{document}

\title{Multi-Stream Opportunistic Network Decoupling: Relay Selection and Interference Management}
\author{Huifa~Lin,~\IEEEmembership{Member,~IEEE,} Won-Yong~Shin,~\IEEEmembership{Senior~Member,~IEEE}, and~Bang~Chul~Jung,~\IEEEmembership{Senior~Member,~IEEE}
\IEEEcompsocitemizethanks{\IEEEcompsocthanksitem H. Lin is with the Communications \& Networking Laboratory, Dankook University, Yongin 16890, Republic of Korea. \protect\\
Email: huifa.lin.dr@ieee.org.
\IEEEcompsocthanksitem W.-Y. Shin is with the Department of Computer Science and Engineering, Dankook University, Yongin 16890, Republic of Korea. \protect\\
Email: wyshin@dankook.ac.kr.
\IEEEcompsocthanksitem B. C. Jung is with the Department of Electronics Engineering, Chungnam National University, Daejeon 34134, Republic of Korea. \protect\\
Email: bcjung@cnu.ac.kr.}
}
\maketitle

\begin{abstract}
We study multi-stream transmission in the $K \times N \times K$ channel with \emph{interfering relay nodes}, consisting of $K$ \emph{multi-antenna} source--destination (S--D) pairs and $N$ single-antenna half-duplex relay nodes between the S--D pairs.
We propose a new achievable scheme operating with partial effective channel {\em gain}, termed \emph{multi-stream opportunistic network decoupling (MS-OND)}, which achieves the optimal degrees of freedom (DoF) under a certain relay scaling law.
Our protocol is built upon the conventional OND that leads to {\em virtual full-duplex} mode with one data stream transmission per S--D pair, generalizing the idea of OND to multi-stream scenarios by leveraging relay selection and interference management.
Specifically, two subsets of relay nodes are opportunistically selected using alternate relaying in terms of producing or receiving the minimum total interference level.
For interference management, each source node sends $S \,(1 \le S \le M)$ data streams to selected relay nodes with \emph{random beamforming} for the first hop, while each destination node receives its desired $S$ streams from the selected relay nodes via {\em opportunistic interference alignment} for the second hop, where $M$ is the number of antennas at each source or destination node.
Our analytical results are validated by numerical evaluation.
\end{abstract}


\section{Introduction}
Internet of Things (IoT) has been emerging as a promising technology that integrates the physical world into computer-based systems \cite{ercan2017rf}.
Recent developments of the IoT have also spurred research and standardization efforts on massive machine type communications (mMTC) in the fifth generation (5G) wireless networks \cite{agiwal2016next}.
In such wireless networks, a massive number of devices with low energy and low cost can be deployed, e.g., connection density of $1 \times 10^6$ devices per km$^2$ in urban areas may be necessary \cite{3gpp2016study}, where the half-duplex and single-antenna configuration is preferable \cite{akpakwu2017survey}.
Thus, it would be important to design an effective protocol that guarantees satisfactory performance even under such low-cost requirements on the devices.

\subsection{Previous Work}
Interference management has been taken into account as one of the most challenging and important issues in wireless multiuser communications~\cite{Lee2014a}.
While it has been elusive to characterize the Shannon-theoretic capacity of interference channels, interference alignment (IA) was proposed for fundamentally solving the interference problem among multiple communication pairs \cite{maddah2008communication, cadambe2008interference}.
It was shown that the IA scheme in \cite{cadambe2008interference} can achieve the optimal degrees of freedom (DoF), which is equal to $K/2$, in the $K$-user interference channel with time-varying channel coefficients.
Interference management schemes based on IA have been further developed and analyzed in various wireless network environments such as multiple-input multiple-output (MIMO) interference networks \cite{gomadam2011distributed, gou2010degrees}, $X$ networks \cite{jafar2008degrees}, and cellular networks \cite{suh2008interference, motahari2014real, jung2011opportunistic, jung2012opportunistic, yang2013opportunistic, yang2017opportunistic}.

Recently, the $K$-user \emph{two-hop relay-aided} interference channel (also known as the $K \times N \times K$ channel), which consists of $K$ \mbox{source--destination (S--D)} pairs and $N$ helping relay nodes located between the S--D pairs, has received a great deal of attention from academia~\cite{gou2012aligned, shomorony2014degrees, zanella2017relay}.
The $K \times N \times K$ channel is more challenging than the $K$-user interference channel because interference management and cooperative relaying operations that are coupled with each other need to be carefully conducted.
In the $2 \times 2 \times 2$ interference channel, as a special case of the $K \times N \times K$ channel, it was shown that interference neutralization achieves the optimal DoF~\cite{gou2012aligned}.
In addition, aligned network diagonalization was proposed for the general $K \times N \times K$ channel to achieve the optimal DoF~\cite{shomorony2014degrees}.
However, it was assumed in~\cite{gou2012aligned, shomorony2014degrees, zanella2017relay} that relay nodes are full-duplex and/or there is no interfering signal among relay nodes.

On the other hand, there have been extensive studies on how to exploit the {\it multiuser diversity} gain in single-cell downlink scenarios when the number of users is sufficiently large by introducing opportunistic scheduling \cite{knopp1995information}, opportunistic beamforming \cite{viswanath2002opportunistic}, and random beamforming~(RBF) \cite{sharif2005capacity}.
For multi-cell downlink scenarios, multi-cell RBF schemes were proposed in~\cite{shin2012network, nguyen2013multi}.
Moreover, a joint design of IA-enabled beamforming and opportunistic scheduling, called \emph{opportunistic interference alignment (OIA)}, has been proposed in multi-cell uplink or downlink networks~\cite{jung2011opportunistic, jung2012opportunistic, yang2013opportunistic, yang2017opportunistic}.
Even without centralized controlling, the benefits of opportunistic transmission were also examined in slotted ALOHA-based random access networks~\cite{qin2006distributed, adireddy2005exploiting, lin2017multi, lin2017MAORA}.
By applying opportunism to cooperative communications, various techniques such as opportunistic two-hop relaying \cite{cui2009opportunistic, lin2016cognitive} and opportunistic routing \cite{shin2013parallel, gao2015forwarding, so2017load} were investigated.
As for the $K \times N \times K$ channel having \emph{interfering relay nodes}, \emph{opportunistic network decoupling (OND)} was recently proposed while showing that $K$ DoF is asymptotically achieved even in the presence of inter-relay interference when $N$ is beyond a certain value~\cite{shin2017opportunistic}.
In the OND protocol, two sets of relay nodes are selected among total $N$ relay candidates to alternatively receive signals from source nodes or forward signals to destination nodes in each time slot, thus realizing the virtual full-duplex mode.
The two relay sets are opportunistically selected in the sense that both the interference among S--D pairs and the interference among relay nodes are effectively controlled.
The OND protocol in~\cite{shin2017opportunistic} would be feasible in practice in the sense that not only it effectively manages the inter-relay interference unlike the studies in~\cite{gou2012aligned, shomorony2014degrees, zanella2017relay} but also the network operates in virtual full-duplex mode even with half-duplex relay nodes.

Meanwhile, to deal with self-interference that is generally far stronger than the signal of interest in full-duplex systems \cite{bharadia2013full}, several self-interference cancellation (SIC) techniques have been developed. Examples include the sum-rate optimization for full-duplex multi-antenna relaying systems under limited dynamic range \cite{day2012full, kim2013distributed}.

\subsection{Main Contributions} \label{sec_introduction}
The prior work in~\cite{shin2017opportunistic} basically deals with single-stream transmission for each S--D pair since a single antenna is assumed to be deployed at each source and destination node.
With the increasing number of antennas at mobile terminals in wireless communication systems, a natural question arises as follows: how can one successfully deliver multiple data streams for each \emph{multi-antenna} S--D pair by fully exploiting the multiuser diversity gain in fading channels?
We attempt to answer this fundamental question in this paper.
As an extension of the single-antenna configuration in \cite{shomorony2014degrees, shin2017opportunistic}, we consider the multi-antenna $K \times N \times K$ channel with $N$ single-antenna interfering \emph{half-duplex} relay nodes operating in time-division duplex (TDD) mode, where each of $K$ source and destination nodes is equipped with $M$ antennas and each source node sends $S~(1\le S \le M)$ data streams.
Extension to the \emph{multi-stream} scenario is not straightforward since more challenging and sophisticated interference management and relay selection strategies are accompanied under the channel model.
In particular, we need to elaborately handle the inter-stream interference among multiple spatial streams in each S--D pair, in addition to the inter-pair interference and inter-relay interference that have appeared in the single-antenna $K\times N\times K$ channel~\cite{shin2017opportunistic}.

In this paper, we propose a \emph{multi-stream OND (MS-OND)} protocol operating in a fully distributed manner only with partial effective channel \emph{gain} information at the transmitters.
Typical application scenarios of the proposed MS-OND protocol include mMTC and IoT in the 5G wireless networks, where a massive number of low-cost devices with the {\em half-duplex and single-antenna configuration} can be deployed, providing potentially strong supports as candidate relay nodes \cite{akpakwu2017survey}.
Based upon the single-stream OND protocol in~\cite{shin2017opportunistic}, MS-OND is designed by further leveraging both interference management and relay selection techniques. To be specific, two subsets of relay nodes among $N$ relay candidates are opportunistically selected while using {\em alternate relaying} in terms of generating or receiving the minimum total interference level (TIL), which eventually enables our system to operate in {\em virtual full-duplex} mode. Furthermore, for interference management, our protocol intelligently integrates {\em RBF} for the first hop and {\em OIA} for the second hop into the comprehensive network decoupling framework.
Such a protocol integration is a challenging task since it involves various techniques across different domains such as scheduling, beamforming, and interference management.
As our main result, it is shown that in a high signal-to-noise ratio (SNR) regime, the proposed MS-OND protocol achieves $SK$ DoF provided that the number of relay nodes, $N$, scales faster than $\text{SNR}^{3SK-S-1}$, which is the minimum number of relay nodes required to guarantee the DoF achievability.

Our main contributions are fourfold and summarized as follows:
\begin{itemize}
  \item For the multi-antenna $K \times N \times K$ channel with interfering relay nodes, we introduce a general OND framework, which enables each S--D pair to perform multi-stream communications by incorporating the notion of RBF and OIA techniques into the protocol design.
  \item Under the channel model, we completely analyze the achievable DoF under a certain relay scaling condition and the decaying rate of the TIL. Furthermore, the MS-OND protocol is shown to asymptotically approach the cut-set upper bound on the DoF.
  \item Our analytical results (i.e., the relay scaling law required to achieve a given DoF) are numerically validated through extensive computer simulations.
  \item We also perform extensive computer simulations in {\em finite} system parameter regimes to show when the MS-OND protocol with alternate relaying is superior in practice.
\end{itemize}

\subsection{Organizations}
The rest of this paper is organized as follows.
In Section \ref{sec_model}, we describe the system and channel models.
In Section \ref{sec_method}, the proposed MS-OND protocol is described.
Section \ref{sec_analysis} presents analysis on both the achievable DoF and the decaying rate of the TIL.
Numerical results for the proposed MS-OND protocol are provided in Section \ref{sec_results}.
Finally, we summarize the paper with some concluding remarks in Section \ref{sec_conclusion}.

\subsection{Notations}
\begin{table}[t]
 \centering
 \caption{Summary of notations.} \label{table_notations}
 \renewcommand{\arraystretch}{1.6}
 \begin{tabular}{ | c | c | }
 \hline
  Notation & Description \\
  \hline
  $K$ & number of S--D pairs \\
  \hline
  $N$ & number of relay nodes \\
  \hline
  $M$ & number of antennas at each S--D pair \\
  \hline
  $S$ & number of data streams per S--D pair \\
  \hline
  $\mathcal{S}_k$ & $k$th source node \\
  \hline
  $\mathcal{D}_k$ & $k$th destination node \\
  \hline
  $\mathcal{R}_n$ & $n$th relay node \\
  \hline
  $\mathbf{h}_{nk}^{(1)}$ & channel coefficient vector from $\mathcal{S}_k$ to $\mathcal{R}_n$ \\
  \hline
  $\mathbf{h}_{kn}^{(2)}$ & channel coefficient vector from $\mathcal{R}_n$ to $\mathcal{D}_k$ \\
  \hline
  $h_{mn}^{(r)}$ & channel coefficient between $\mathcal{R}_n$ and $\mathcal{R}_m$ \\
  \hline
  $\mathbf{\Pi}_b (b = \{1, 2\})$ & two selected relay sets \\
  \hline
  $L_{n, (k,s)}^{\mathbf{\Pi}_1}$ & scheduling metric of the first relay set $\mathbf{\Pi}_1$\\
  \hline
  $L_{n, (k,s)}^{\mathbf{\Pi}_2}$ & scheduling metric of the second relay set $\mathbf{\Pi}_2$\\
  \hline
  $\textsf{DoF}_\text{total}$ & total number of DoF \\
  \hline
  \end{tabular}
\end{table}
Throughout this paper, $\mathbb{C}$, $\mathbb{E}[\cdot]$, and $\lceil \cdot \rceil$ indicate the field of complex numbers, the statistical expectation, and the ceiling operation, respectively.
Unless otherwise stated, all logarithms are assumed to be to the base $2$.
We use the following asymptotic notations: $f(x) = O(g(x))$ means that there exist constants $C$ and $c$ such that $f(x) \leq Cg(x)$ for all $x>c$; $f(x) = \Omega(g(x))$ if $g(x) = O(f(x))$; and $f(x) = \omega(g(x))$ means that $\lim_{x \to \infty}{\frac{g(x)}{f(x)} = 0}$ \cite{knuth1976big}.
Moreover, Table \ref{table_notations} summarizes the notations used throughout this paper. Some notations will be more precisely defined in the following sections, as we introduce our channel model and achievability results.

\section{System and Channel Models} \label{sec_model}
We consider the multi-antenna $K \times N \times K$ channel with interfering relay nodes, where each source or destination node is equipped with $M$ antennas while each relay node is equipped with a single antenna.\footnote{We do not assume to equip multiple antennas at each relay node since it does not further improve the DoF and may cause the space limitation as relay nodes are treated as small-size sensors.}
We assume that there exists no direct communication path between each S--D pair as the source and destination nodes are geographically far apart.
Each source node sends $S~(1 \leq S \leq M)$ independent data streams to the corresponding destination node through $2S$ relay nodes.
There are two relay sets composed of $2SK$ relay nodes, where each relay set is opportunistically selected out of $N$ relay candidates (which will be specified in Section \ref{subsec_selection}).\footnote{As mentioned in Section \ref{sec_introduction}, it is assumed to deploy a massive number of devices as potential relay nodes in mMTC or IoT wireless networks, which are the target application scenarios of our MS-OND protocol.}
Each relay node is assumed to operate in half-duplex mode and to fully decode, re-encode, and retransmit the source data, i.e., to employ decode-and-forward (DF) relaying.
The relay nodes are assumed to interfere with each other when sending data to the belonging destination nodes.
We assume that each node (either a source node or a relay node) has an average transmit power constraint $P$.

Let $\mathcal{S}_k$, $\mathcal{D}_k$, and $\mathcal{R}_n$ denote the $k$th source node, the corresponding $k$th destination node, and the $n$th relay node, respectively, where $k \in \{ 1, 2, \cdots, K \}$ and $n \in \{ 1, 2, \cdots, N \}$.
The channel coefficient vector from the source $\mathcal{S}_k$ to the relay $\mathcal{R}_n$, corresponding to the first hop, is denoted by $\mathbf{h}_{nk}^{(1)} \in \mathbb{C}^{M \times 1}$; the channel coefficient vector from the relay $\mathcal{R}_n$ to the destination $\mathcal{D}_k$, corresponding to the second hop, is denoted by $\mathbf{h}_{kn}^{(2)} \in \mathbb{C}^{1 \times M}$; and the channel coefficient between two relay nodes $\mathcal{R}_n$ and $\mathcal{R}_m$ is denoted by $h_{mn}^{(r)}$.
All the channel coefficients are assumed to be Rayleigh, having zero-mean and unit variance, and to be independent independent and identically distributed (i.i.d.) over different $k, n, m$, and the hop index.
We assume the block-fading model, i.e., the channels are constant during one block, consisting of one initialization phase and $L$ data transmission time slots, and change to new independent values for every block.
In our work, we assume that partial channel {\em gain} information (i.e., channel gains that can be estimated via pilot signaling) is only available at the transmitters.

\section{Proposed MS-OND Protocol} \label{sec_method}
\begin{figure}
\begin {center}
\epsfig{file=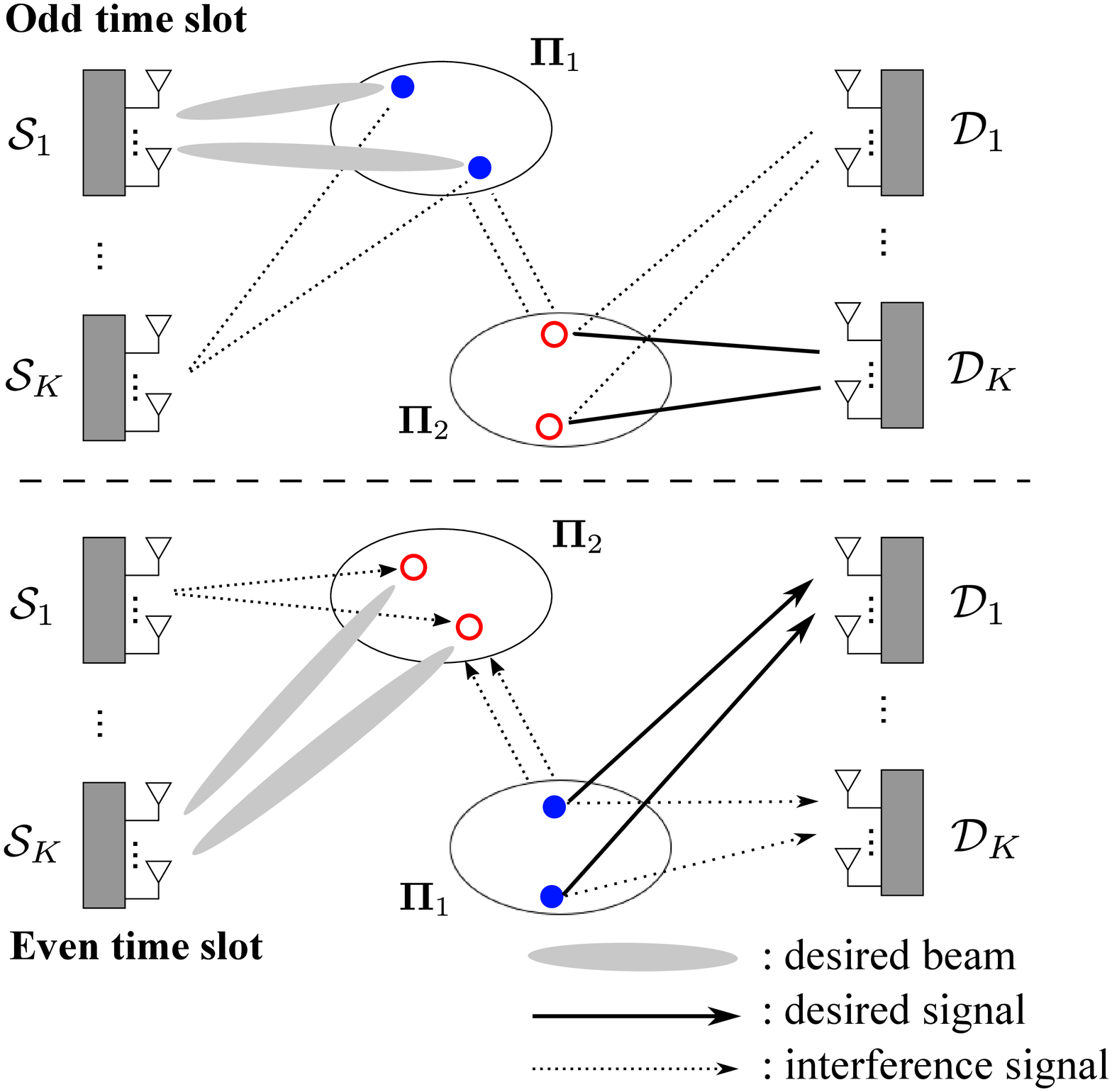, width=0.73\hsize}
\end {center}
\caption{Illustration of the MS-OND protocol operation in the $K\times N\times K$ channel with interfering relay nodes when each S--D pair equips multiple antennas.}
\label{fig_mond}
\end{figure}
%overall procedure
In this section, we elaborate on the MS-OND protocol as an achievable scheme that guarantees the optimal DoF for the multi-antenna $K \times N \times K$ channel with inter-relay interference, where $2SK$ relay nodes among $N$ relay candidates are opportunistically selected for data reception and forwarding in the sense of generating or receiving a sufficiently small amount of interference.
Furthermore, we describe two opportunistic transmission techniques including RBF for the first hop and OIA for the second hop that are intelligently integrated into our network decoupling framework.

\subsection{Overall Procedure}
For the sake of explanation, we assume that the number of data transmission time slots, $L$, is odd.
The overall procedure of the MS-OND protocol is described as follows.

\textbf{(a) Initialization phase}: The source nodes generate and broadcast RBF vectors for the first hop.
More specifically, the source $S_k$ generates $M$ RBF vectors according to the isotropic distribution \cite{hassibi2002multiple} that are constructed by an $M \times M$ unitary matrix $\mathbf{V}_k = \left[ \mathbf{v}_k^{(1)}, \mathbf{v}_k^{(2)}, \cdots,  \mathbf{v}_k^{(M)} \right]$.
Here, $ \mathbf{v}_k^{(i)} \in \mathbb{C}^{M \times 1}$ is the $i$th RBF vector of the source $\mathcal{S}_k$, where $i \in \{1, 2, \cdots, M\}$.
The destination nodes generate and broadcast their interference space so that OIA is employed for the second hop.
More specifically, the destination $D_k$ generates the interference space denoted by $\mathbf{Q}_k = \left[ \mathbf{q}_k^{(1)}, \cdots,  \mathbf{q}_k^{(M-S)} \right]$, where $ \mathbf{q}_k^{(i)} \in \mathbb{C}^{M \times 1}$ is the orthonormal basis and $S \in \{1, 2, \cdots, M\}$.
The corresponding null space of the destination $D_k$, indicating the signal space, is defined as $\mathbf{U}_k = \left[ \mathbf{u}_k^{(1)}, \cdots,  \mathbf{u}_k^{(S)} \right] \triangleq \text{null}(\mathbf{Q}_k)$, where $ \mathbf{u}_k^{(i)} \in \mathbb{C}^{M \times 1}$ is the orthonormal basis.
For the space generation, we choose $M-S$ columns of the left or right singular matrix of any $M \times M$ matrix as $\mathbf{Q}_k$ and choose the rest $S$ columns as $\mathbf{U}_k$. If $S=M$, then $\mathbf{U}_k$ can be any orthogonal matrix.

Two relay sets $\mathbf{\Pi}_1$ and $\mathbf{\Pi}_2$ among $N$ relay candidates, each of which consists of $SK$ relay nodes, are selected to alternately receive and forward $S K$ independent data streams for $K$ S--D pairs so that each S--D pair is assisted by $2S$ relay nodes.
The relay sets $\mathbf{\Pi}_1$ and $\mathbf{\Pi}_2$ are denoted by
\begin{equation} \label{ } \nonumber
\begin{split}
\mathbf{\Pi}_1 = \{&\pi_1(1,1), \pi_1(1,2), \cdots, \pi_1(1,S), \cdots,\\
&\pi_1(K,1), \pi_1(K,2), \cdots, \pi_1(K,S)\}
\end{split}
\end{equation}
and
\begin{equation} \label{ } \nonumber
\begin{split}
\mathbf{\Pi}_2 = \{&\pi_2(1,1), \pi_2(1,2), \cdots, \pi_2(1,S), \cdots,\\
&\pi_2(K,1), \pi_2(K,2), \cdots, \pi_2(K,S)\},
\end{split}
\end{equation}
respectively.
Here, $\pi_b(k,s)$ denotes the index of the relay node who serves the $s$th data stream of the $k$th S--D pair in the relay set $\mathbf{\Pi}_b$, where $b \in \{1, 2\}$, $k \in \{1, 2, \cdots, K\}$, and $s \in \{1, 2, \cdots, S\}$.
The rest $N - 2SK$ relay nodes remain idle during all the time slots.

\textbf{(b) Odd time slot $l_o \in \{1, 3, \cdots, L\}$}: As shown in Fig. \ref{fig_mond}, in virtue of RBF, each source node transmits its $S$ encoded symbols along with $S$ spatial beams to $S$ relay nodes in the relay set $\mathbf{\Pi}_1$.
For instance, the source $\mathcal{S}_k$ transmits symbols $x_{k,1}^{(1)}{(l_o)}, \cdots, x_{k,S}^{(1)}{(l_o)}$ on $S$ spatial beams, where $k \in \{1, 2, \cdots, K\}$.
The relay $\mathcal{R}_{\pi_1(k, s)}$ receives the desired symbol $x_{k, s}^{(1)}{(l_o)}$ on the $s$th spatial beam of the source $\mathcal{S}_k$, where $s \in \{1, 2, \cdots, S\}$.\footnote{Detailed data transmission process will be described in Section \ref{subsec_DT}.}
Meanwhile, the relay nodes in the relay set $\mathbf{\Pi}_2$ forward the symbols received from the source nodes in the previous time slot by using DF relaying at the same time.
Note that in the first time slot (i.e., $l_o=1$), the relay nodes in the relay set $\mathbf{\Pi}_2$ remain idle since there is no symbol to forward.
In the last time slot (i.e., $l_o=L$), the S--R transmission is not required, and thus all the source nodes and the relay nodes in the relay set $\mathbf{\Pi}_1$ remain idle.

\textbf{(c) Even time slot $l_e \in \{2, 4, \cdots, L-1\}$}:
The source $\mathcal{S}_k$ transmits $S$ encoded symbols $x_{k,1}^{(1)}{(l_e)}, \cdots, x_{k,S}^{(1)}{(l_e)}$ to $S$ relay nodes in the relay set $\mathbf{\Pi}_2$.
The relay nodes in the relay set $\mathbf{\Pi}_1$ forward their re-encoded symbols to the intended destination nodes at the same time.
For instance, relay nodes $\mathcal{R}_{\pi_1(k,1)}, \cdots, \mathcal{R}_{\pi_1(k,S)}$ transmit the symbols $x_{\pi_1(k,1)}^{(2)}{(l_e-1)}, \cdots, x_{\pi_1(k,S)}^{(2)}{(l_e-1)}$ to the destination $\mathcal{D}_k$ while the interfering signals to other destination nodes are opportunistically aligned to their interference space.
If the relay $R_{\pi_1(k,s)}$ in the relay set $\mathbf{\Pi}_1$ successfully decodes its desired symbol, then $x_{\pi_1(k, s)}^{(2)}{(l_e-1)}$ is the same as $x_{k, s}^{(1)}{(l_e-1)}$.

\textbf{(d) Reception at the destination nodes}:
The relay sets $\mathbf{\Pi}_1$ and $\mathbf{\Pi}_2$ alternately operate in receive and transmit modes in the odd and even time slots, respectively.
The destination nodes operate in receive mode for all the time slots except the first time slot, while decoding the symbols forwarded by either the relay set $\mathbf{\Pi}_1$ (in the even time slots) or the relay set $\mathbf{\Pi}_2$ (in the odd time slots) by adopting zero-forcing (ZF) detection.\footnote{It is worth noting that more sophisticated detection methods (e.g., minimum mean square error detection) can also be employed at the receivers since our MS-OND protocol is detector-agnostic. Even if employing a more sophisticated detection method could further improve the sum-rate performance of our MS-OND protocol, we adopt ZF detection, which is sufficient to achieve the optimal DoF.}
Note that for each destination node, the received interference from the relay nodes other than its own assisting relay nodes can be well confined in virtue of OIA for the second hop as the number of relay nodes is sufficiently large.

\subsection{Relay Set Selection} \label{subsec_selection}
In this section, we describe how to select the two relay sets $\mathbf{\Pi}_1$ and $\mathbf{\Pi}_2$.
By exploiting the multiuser diversity gain in fading channels, relay nodes are opportunistically selected in the sense of generating or receiving the minimum sum amount of the following threes types of interference: i) interference from other spatial beams during the S--R transmission; ii) interference leakage to other destination nodes during the R--D transmission; and iii) interference between two relay sets.

\subsubsection{The First Relay Set Selection}
Let us first focus on selecting the relay set $\mathbf{\Pi}_1$ from $N$ relay candidates, which operates in receive and transmit modes in odd and even time slots, respectively.
For every initialization period, it is possible for the relay $\mathcal{R}_n$ to acquire a part of {\em effective channel gain} information via pilot signaling sent from all the source and destination nodes due to the channel reciprocity of TDD systems before data transmission, where $n \in \{1, 2, \cdots, N\}$ and $k \in \{1, 2, \cdots, K\}$ (which will be specified later).
When the relay $\mathcal{R}_n$ is assumed to serve the $s$th data stream of the $k$th S--D pair, it examines both i) how much interference is received from other spatial beams created by RBF for the first hop, including the interference from other source nodes and the interference from other $S-1$ spatial beams of the source $\mathcal{S}_k$; and ii) how much interference leakage is generated by itself to other destination nodes via OIA for the second hop.
Then, the relay $\mathcal{R}_n$ computes the following scheduling metric $L_{n, (k,s)}^{\mathbf{\Pi}_1}$:
\begin{equation} \label{eq_metric_set1}
\begin{split}
L_{n, (k,s)}^{\mathbf{\Pi}_1} &= \sum_{t=1 \atop t \neq s}^{S}{\left|\mathbf{v}_{k}^{(t)T} \mathbf{h}_{nk}^{(1)} \right|^2} +
\sum_{j=1 \atop j \neq k}^{K}{  \sum_{t=1}^{S}{ \left|\mathbf{v}_{j}^{(t)T} \mathbf{h}_{nj}^{(1)} \right|^2 }}\\ &+ \sum_{j=1 \atop j \neq k}^{K}{\left\| \text{Proj}_{\bot \mathbf{Q}_j} \left( \mathbf{h}_{jn}^{(2)} \right) \right\|^2},
\end{split}
\end{equation}
where $n \in \{ 1, 2, \cdots, N \}$, $k \in \{ 1, 2, \cdots, K \}$, and $s \in \{ 1, 2, \cdots, S \}$.
Here, $\text{Proj}_{\bot \mathbf{Q}_j} \left( \mathbf{h}_{jn}^{(2)} \right) \triangleq \mathbf{U}_j^H \mathbf{h}_{jn}^{(2)}$ denotes the orthogonal projection of $\mathbf{h}_{jn}^{(2)}$ onto $\mathbf{U}_j$.
Thus, the last term $\sum_{j=1 \atop j \neq k}^{K}{\left\| \text{Proj}_{\bot \mathbf{Q}_j} \left( \mathbf{h}_{jn}^{(2)} \right) \right\|^2}$ in (\ref{eq_metric_set1}) indicates the sum of interference leakage links generated by the relay $\mathcal{R}_n$ to other $K-1$ destination nodes, which can also be estimated at each relay node as all the destination nodes send pilot signaling multiplied by their null space.
In this selection stage, we aim to find the relay set $\mathbf{\Pi}_1$ leading to negligibly small values of $L_{n,(k,s)}^{\mathbf{\Pi}_1}$.
We also remark that the first and second terms $ \sum_{t=1 \atop t \neq s}^{S}{\left|\mathbf{v}_{k}^{(t)T} \mathbf{h}_{nk}^{(1)} \right|^2}$ and $\sum_{j=1 \atop j \neq k}^{K}{  \sum_{t=1}^{S}{ \left|\mathbf{v}_{j}^{(t)T} \mathbf{h}_{nj}^{(1)} \right|^2 }}$ in (\ref{eq_metric_set1}) denote the sum of interference links from other spatial beams of the source $\mathcal{S}_k$ and from other source nodes, respectively, which can be estimated at each relay node as all the source nodes send pilot signaling multiplied by their RBF vectors.
Note that inter-relay interference cannot be computed when the relay set $\mathbf{\Pi}_1$ is selected because neither $\mathbf{\Pi}_1$ nor $\mathbf{\Pi}_2$ is available in this phase and it is sufficient to consider the inter-relay interference when the relay set $\mathbf{\Pi}_2$ is selected.

\subsubsection{The Second Relay Set Selection}
Now let us turn to selecting the relay set $\mathbf{\Pi}_2$ from the remaining $N - SK$ relay candidates, which operates in receive and transmit modes in even and odd time slots, respectively.
After the selection of the relay set $\mathbf{\Pi}_1$, it is possible for the relay $\mathcal{R}_n \in \{1, \cdots, N\} \setminus \mathbf{\Pi}_1$ to compute the sum of inter-relay interference links generated by the relay set $\mathbf{\Pi}_1$, denoted by $\sum_{j=1}^{K}{ \sum_{t=1}^{S}{ \left| h_{n \pi_1(j, t)}^{(r)} \right|^2 } }$, which can be estimated through pilot signaling sent from the relay nodes belonging to the relay set $\mathbf{\Pi}_1$.
When the relay $\mathcal{R}_n$ is assumed to serve the $s$th data stream of the $k$th S--D pair, it examines both i) how much interference is received from the undesired spatial beams created by RBF for the first hop and from the relay set $\mathbf{\Pi}_1$; and ii) how much interference leakage is generated by itself to other destination nodes via OIA for the second hop.
Then, the relay $\mathcal{R}_n$ computes the following scheduling metric $L_{n, (k,s)}^{\mathbf{\Pi}_2}$, termed \emph{TIL} in this paper:
\begin{equation} \label{eq_metric_set2}
L_{n, (k,s)}^{\mathbf{\Pi}_2} = L_{n, (k,s)}^{\mathbf{\Pi}_1} + \sum_{j=1}^{K}{ \sum_{t=1}^{S}{ \left| h_{n \pi_1(j,t)}^{(r)} \right|^2 } },
\end{equation}
where $n \in \{ 1, 2, \cdots, N \}$ and $k \in \{ 1, 2, \cdots, K \}$.
As long as the relay nodes in the relay set $\mathbf{\Pi}_2$ having a sufficiently small amount of the TIL are selected, the sum of inter-relay interference links received at the relay nodes in the relay set $\mathbf{\Pi}_1$ becomes sufficiently small due to the channel reciprocity.
As a result, our system is now capable of operating in {\em virtual full duplex} mode even with half-duplex relay nodes since the relay nodes in the relay set $\mathbf{\Pi}_1$ are in receive mode with almost no inter-relay interference when the relay nodes in the relay set $\mathbf{\Pi}_2$ are in transmit mode, or vice versa.
In this selection stage, the relay set $\mathbf{\Pi}_2$ is found in the sense of having negligibly small values of $L_{n,(k,s)}^{\mathbf{\Pi}_2}$.

\subsubsection{Implementation Based on Distributed Timers} \label{subsec_timer}
After the two scheduling metrics $L_{n, (k,s)}^{\mathbf{\Pi}_1}$ and $L_{n, (k,s)}^{\mathbf{\Pi}_2}$ are computed at each relay node, a crucial question that we would like to raise is how to select relay nodes in a distributed manner.
To answer this question, a timer-based method can be adopted similarly as in \cite{shin2017opportunistic, bletsas2006simple}, which operates based on the exchange of a short duration CTS (Clear to Send) message transmitted by each destination node who finds its desired relay node.\footnote{The reception of a CTS message that is transmitted from a certain destination node triggers the initial timing process at each relay node. Therefore, no explicit timing synchronization protocol is required among relay nodes \cite{bletsas2006simple}. Moreover, it is worth noting that the overhead of relay selection is a small fraction of one transmission block with small collision probability \cite{bletsas2006simple}. Note that our relay selection procedure is performed sequentially over all the S--D pairs and the selected relay node for a data stream of one certain S--D pair is not allowed to participate in the selection process for another data stream of the belonging S--D pair or another S-D pair.}
Such a timer-based selection method would be considerably suitable in distributed systems in the sense that information exchange among relay nodes can be minimized.
The selection process of the relay set $\mathbf{\Pi}_1$ first begins.
It is straightforward that the selection process of the relay set $\mathbf{\Pi}_2$ can be performed in a similar manner.

\textbf{(a) The selection process of $\mathbf{\Pi}_1$}:
At the beginning of every scheduling period, the relay $\mathcal{R}_n$ for $n \in \{1, 2, \cdots, N\}$ computes $SK$ scheduling metrics for all the data streams, consisting of $S$ scheduling metrics $L_{n, (k,1)}^{\mathbf{\Pi}_1}, \cdots, L_{n, (k,S)}^{\mathbf{\Pi}_1}$ for each source $\mathcal{S}_k$, where $k \in \{1, 2, \cdots, K\}$.
Then, the relay $\mathcal{R}_n$ starts with $SK$ timers whose initial values are proportional to the $SK$ scheduling metrics.
Thus, over the whole network, there are $NSK$ timers prepared for the selection of $\mathbf{\Pi}_1$.
Suppose that a timer of the relay $\mathcal{R}_{\pi_1(\hat{k}, \hat{s})}$ having the smallest value in the network, denoted by $L_{\pi_1(\hat{k}, \hat{s}), (\hat{k}, \hat{s})}^{\mathbf{\Pi}_1}$, expires first.
Then, the relay $\mathcal{R}_{\pi_1(\hat{k}, \hat{s})}$ transmits a short duration RTS (Ready to Send) message, signaling its presence, to other $N - 1$ relay nodes.
The RTS message is composed of $\ceil{\log_2{SK}}$ bits to indicate the corresponding data stream of a certain S--D pair to be served.
Subsequently, the following actions are performed by relay nodes: i) the relay $\mathcal{R}_{\pi_1(\hat{k}, \hat{s})}$ clears all its remained $SK - 1$ timers and keeps idle thereafter during the scheduling period; and ii) after receiving the RTS message sent by the relay $\mathcal{R}_{\pi_1(\hat{k}, \hat{s})}$, other $N - 1$ relay nodes clear their timers corresponding to the data stream reserved by the relay $\mathcal{R}_{\pi_1(\hat{k}, \hat{s})}$.
Now, there exist $(N-1)(SK-1)$ timers left by $N-1$ relay candidates.
Each relay candidate continues to listen to the RTS message from other relay nodes while waiting for its own timer(s) to expire.
If another relay node sends the second RTS message of $\ceil{\log_2{(SK-1)}}$ bits in order to declare its presence, then it is selected to communicate with the corresponding data stream of the S--D pair to be served.
Each relay candidate keeps on checking the number of data streams reserved for each S--D pair.
More precisely, if all $S$ data streams of the $k$th S--D pair are reserved, then each relay candidate clears all its timers for the source $\mathcal{S}_k$.
After $SK$ RTS messages are sent out in consecutive order, yielding no RTS collision with high probability, the selection of $\mathbf{\Pi}_1$ is terminated.
When an RTS collision occurs (i.e., two or more relay nodes have exactly the same value of the scheduling metric), none of the relay nodes is selected.
The network waits for such a relay node whose timer will expire next.

\textbf{(b) The selection process of $\mathbf{\Pi}_2$}:
The selection process of $\mathbf{\Pi}_2$ begins after completion of selecting $\mathbf{\Pi}_1$.
By changing the scheduling metric to $L_{n,(k,s)}^{\mathbf{\Pi}_2}$, $\mathbf{\Pi}_2$ can be selected from relay candidates $\mathcal{R}_n \in \{1, \cdots, N\} \setminus \mathbf{\Pi}_1$ by applying the distributed timer-based method as shown above.
It is worth noting that for selection of the two relay sets $\mathbf{\Pi}_1$ and $\mathbf{\Pi}_2$, only $2SK \ceil{\log_2{SK}}$ bits could suffice for information exchange during the scheduling period, which would be proportionally negligibly small when $L$ is large.

\subsection{Data Transmission} \label{subsec_DT}
After the selection process of two relay sets $\mathbf{\Pi}_1$ and $\mathbf{\Pi}_2$, each source node ($\mathcal{S}_k$) starts to transmit $S$ data streams to the corresponding destination node ($\mathcal{D}_k$) via $S$ relay nodes belonging to either $\mathbf{\Pi}_1$ or $\mathbf{\Pi}_2$.
Without loss of generality, we focus on each odd time slot, i.e., $l_o=\{1,3,\cdots,L\}$.
Let us first explain the basic operation of reception and transmission for the relay nodes in $\mathbf{\Pi}_1$.
For the first hop, the relay nodes $\mathcal{R}_{\pi_1(k,1)}, \cdots, \mathcal{R}_{\pi_1(k,S)}$ in the relay set $\mathbf{\Pi}_1$ receive the $S$ spatial beams from the source $\mathcal{S}_k$, where each relay node is associated with one beam.
The received signal $y_{\pi_1(k,s)}^{(1)}{(l_o)} \in \mathbb{C}$ at $\mathcal{R}_{\pi_1(k,s)}$ is given by
\begin{equation} \label{eq_rx_signal_hop1}
\begin{split}
y_{\pi_1(k,s)}^{(1)}{(l_o)} &= \sum_{j=1}^{K}{\sum_{t=1}^{S}{ \mathbf{v}_{j}^{(t)T} \mathbf{h}_{\pi_1(k,s) j}^{(1)} x_{j,t}^{(1)}{(l_o)} } }\\
&+ \sum_{j=1}^{K}{\sum_{t=1}^{S}{ h_{\pi_1(s,k) \pi_2(j,t)}^{(r)} x_{\pi_2(j,t)}^{(2)}{(l_o-1)} }}\\
&+ z_{\pi_1(s,k)}^{(1)}{(l_o)},
\end{split}
\end{equation}
where $z_{\pi_1(s,k)}^{(1)}{(\cdot)}$ is the complex additive white Gaussian noise (AWGN), which is i.i.d. over parameters $s$, $k$, and $l_o$, and has zero-mean and variance $N_0$.
For the second hop, the relay nodes in $\mathbf{\Pi}_1$ forward the re-encoded symbols to the corresponding destination nodes by employing the DF relaying, where the received signal is fully decoded, buffered, and re-encoded, e.g., the relay node $\mathcal{R}_{\pi_1(k,s)}$ first recovers $x_{k,s}^{(1)}(l_o)$ in the slot $l_o$ and forwards this signal to the destination node in the next slot $l_{o+1}$.\footnote{We do not deal with buffering issues at the relay nodes because in our MS-OND protocol, each relay node needs only to buffer at most one data symbol.}
The received signal $\mathbf{y}_k^{(2)}(l_o+1) \in \mathbb{C}^{M \times 1}$ at the destination $\mathcal{D}_k$ is given by
\begin{equation} \label{eq_rx_signal_hop2}
\begin{split}
\mathbf{y}_k^{(2)}(l_o+1) = \sum_{j=1}^{K}{ \sum_{t=1}^{S}{ \mathbf{h}_{k \pi_1(j,t)}^{(2)}} x_{\pi_1(j,t)}^{(2)}{(l_o)} } + \mathbf{z}_k^{(2)}(l_o + 1),
\end{split}
\end{equation}
where $\mathbf{z}_k^{(2)}(\cdot) \in \mathbb{C}^{M \times 1}$ denotes the noise vector, each element of which is modeled by an i.i.d. complex AWGN random variable with zero-mean and variance $N_0$.

Likewise, for the relay set $\mathbf{\Pi}_2$, the received signal at the relay $\mathcal{R}_{\pi_2(k,s)}$ for the first hop and the received signal at the destination $\mathcal{D}_k$ for the second hop are given by
\begin{equation} \label{ } \nonumber
\begin{split}
y_{\pi_2(k,s)}^{(1)}{(l_o+1)} &= \sum_{j=1}^{K}{\sum_{t=1}^{S}{ \mathbf{v}_{j}^{(t)T} \mathbf{h}_{\pi_2(k,s) j}^{(1)} x_{j,t}^{(1)}{(l_o+1)} } }\\
&+ \sum_{j=1}^{K}{\sum_{t=1}^{S}{ h_{\pi_2(k,s) \pi_1(j,t)}^{(r)} x_{\pi_1(j,t)}^{(2)}{(l_o)} }}\\
&+ z_{\pi_2(k,s)}^{(1)}{(l_o+1)}
\end{split}
\end{equation}
and
\begin{equation} \label{ } \nonumber
\begin{split}
\mathbf{y}_k^{(2)}(l_o+2) = \sum_{j=1}^{K}{ \sum_{t=1}^{S}{ \mathbf{h}_{k \pi_2(j,t)}^{(2)}} x_{\pi_2(j,t)}^{(2)}{(l_o+1)} } + \mathbf{z}_k^{(2)}(l_o+2),
\end{split}
\end{equation}
respectively.

At each time slot $l\in\{2,3,\cdots,L\}$, by employing OIA for the second hop, the resulting signal vector at the destination $\mathcal{D}_k$ after post-processing is given by
\begin{equation} \label{eq_signal_post_ZF}
\mathbf{r}_k(l) = \left[ r_{k,1}(l), \cdots, r_{k,S}(l) \right]^T \triangleq \mathbf{F}_k^H \mathbf{U}_k^H \mathbf{y}_k^{(2)}(l),
\end{equation}
where $r_{k,s}(l) \in \mathbb{C}$ is the resulting signal corresponding to the $s$th data stream for $s \in \{1, 2, \cdots, S\}$; $\mathbf{U}_k$ indicates the null space of the interference space $\mathbf{Q}_k$ (i.e., the signal space) of the destination $\mathcal{D}_k$ and is multiplied so that inter-pair interference components are aligned at the interference space of the destination $\mathcal{D}_k$; and $\mathbf{F}_k \in \mathbb{C}^{S \times S}$ is a ZF equalizer expressed as
\begin{equation} \label{ } \nonumber
\begin{split}
\mathbf{F}_k &= \left[ \mathbf{f}_{k,1}, \cdots, \mathbf{f}_{k,S} \right] \\
&\triangleq \left( \left[ \mathbf{U}_k^H \mathbf{h}_{k \pi_b(k,1)}^{(2)}, \cdots, \mathbf{U}_k^H \mathbf{h}_{k \pi_b(k,S)}^{(2)} \right]^{-1} \right)^H.
\end{split}
\end{equation}
Here, $\mathbf{f}_{k,s} \in \mathbb{C}^{S \times 1}$ for $s \in \{1, 2, \cdots, S\}$ is the ZF column vector; and $b \in \{1,2\}$ corresponding to the relay sets $\mathbf{\Pi}_1$ and $\mathbf{\Pi}_2$, respectively.

\section{Analysis: DoF and Decaying Rate of TIL} \label{sec_analysis}
In this section, we shall analyze i) the DoF achieved by our proposed MS-OND protocol under a certain relay scaling condition and ii) the decaying rate of the TIL.
The MS-OND protocol without alternate relaying and its achievable DoF are also shown for comparison.

\subsection{DoF Analysis} \label{sec_DoF_analysis}
In this section, using the scaling argument bridging between the number of relay nodes, $N$, and the received SNR (refer to \cite{jung2012opportunistic, yang2013opportunistic, yang2017opportunistic} for the details), we show a lower bound on the DoF achieved by the MS-OND protocol in the multi-antenna $K \times N \times K$ channel with interfering relay nodes and the minimum $N$ required to guarantee the DoF achievability.

The total number of DoF, denoted by $\textsf{DoF}_\text{total}$, is defined as
\begin{equation} \label{ } \nonumber
\textsf{DoF}_\text{total} = \sum_{k=1}^{K}{ \sum_{s=1}^{S}{ \lim_{\textsf{snr} \rightarrow \infty}{ \frac{T_{k,s}(\textsf{snr})}{ \log{\textsf{snr}} } } } },
\end{equation}
where $T_{k,s}(\textsf{snr})$ is the transmission rate for the $s$th data stream of the source $\mathcal{S}_k$ and $\textsf{snr} \triangleq \frac{P}{N_0}$.
Under our MS-OND protocol where $L$ transmission time slots per block are used, the achievable $\textsf{DoF}_\text{total}$ is lower-bounded by
\begin{equation}
\label{eq_DoF_bound1}
\textsf{DoF}_\text{total} \geq \frac{L-1}{L} \sum_{k=1}^{K}{ \sum_{s=1}^{S}{ \sum_{b=1}^{2}{ \left( \lim_{\textsf{snr} \rightarrow \infty} \frac{ \frac{1}{2} \log \left( 1 + \textsf{sinr}_\text{min} \right) }{ \log{\textsf{snr}} } \right) } } },
\end{equation}
where $\textsf{sinr}_\text{min} = \min \left(\textsf{sinr}_{\pi_b(k,s)}^{(1)}, \textsf{sinr}_{k, \pi_b(k,s)}^{(2)} \right)$ is the minimum signal-to-interference-and-noise ratio (SINR) between the two hops.
Here, $\textsf{sinr}_{\pi_b(k,s)}^{(1)}$ denotes the received SINR at the relay $\mathcal{R}_{\pi_b(k,s)}$ and $\textsf{sinr}_{k, \pi_b(k,s)}^{(2)}$ denotes the effective SINR for the $s$th stream at the destination $\mathcal{D}_k$, where $b = \{1, 2\}$ and $k \in \{ 1, 2, \cdots, K \}$.
More specifically, $\textsf{sinr}_{\pi_b(k,s)}^{(1)}$ can be expressed as
\begin{equation} \label{ } \nonumber
\textsf{sinr}_{\pi_b(k,s)}^{(1)} = \frac{P \left| \mathbf{v}_k^{(s)T} \mathbf{h}_{\pi_b(k,s) k}^{(1)} \right|^2}{ N_0 + \mathcal{I}_\text{IB} + \mathcal{I}_\text{IS} + \mathcal{I}_\text{IR}},
\end{equation}
where $\mathcal{I}_\text{IB} = P\sum_{t=1 \atop t \neq s}^{S}{\left| \mathbf{v}_k^{(t)T} \mathbf{h}_{\pi_b(k,t) k}^{(1)} \right|^2}$ is the interference power caused by other generating beams of $\mathcal{S}_k$; $\mathcal{I}_\text{IS} = P \sum_{j=1 \atop j \neq k}^{K}{ \sum_{t=1}^{S}{ \left| \mathbf{v}_j^{(t)T} \mathbf{h}_{\pi_b(j,t) k}^{(1)} \right|^2 } }$ is the interference power from other source nodes; and $\mathcal{I}_\text{IR} = P\sum_{j=1}^{K}{ \sum_{t=1}^{S}{\left| h_{\pi_b(k,s)\pi_{\tilde{b}}(j,t)}^{(r)} \right|^2} }$ is the inter-relay interference.
Let $\tilde{b}$ denote the index of another relay set, resulting in $\tilde{b} = 2$ if $b = 1$ and vice versa.
From (\ref{eq_signal_post_ZF}), the received signal for the $s$th stream at $\mathcal{D}_k$, $r_{k,s}(l)$, is written as
\begin{equation} \label{eq_rx_signal}
\begin{split}
r_{k,s}(l) &= x_{\pi_b(k,s)}^{(2)}(l) + \sum_{j=1 \atop j \neq k}^{K}{ \sum_{t=1}^{S}{ \mathbf{f}_{k,s}^{H} \mathbf{U}^H_k \mathbf{h}_{k \pi_b(j,t)}^{(2)} } x_{\pi_b(j,t)}^{(2)}(l) }\\
&+ \mathbf{f}_{k,s}^{H} \mathbf{\hat{z}}_k(l),
\end{split}
\end{equation}
where $\mathbf{\hat{z}}_k(l) \triangleq \mathbf{U}_k \mathbf{z}_k(l)$.
From (\ref{eq_rx_signal}), $\textsf{sinr}_{k, \pi_b(k,s)}^{(2)}$ is given by
\begin{equation}
\label{ } \nonumber
\textsf{sinr}_{k, \pi_b(k,s)}^{(2)} = \frac{P}{\|\mathbf{f}_{k,s}\|^2 N_0 + \sum\limits_{j=1 \atop j \neq k}^{K}{ \sum\limits_{t=1}^{S}{ \left\| \mathbf{f}_{k,s}^{H} \mathbf{U}^H_k \mathbf{h}_{k \pi_b(j,t)}^{(2)} \right\|^2 } } }.
\end{equation}

We first focus on examining the received SINR values $\textsf{sinr}_{\pi_1(k,s)}^{(1)}$ and $\textsf{sinr}_{\pi_1(k,s)}^{(2)}$ according to each time slot in the relay set $\mathbf{\Pi}_1$'s perspective.
Let us define 
\begin{equation} \label{eq_def_tilde_L1}
\tilde{L}_{\pi_1(k,s), (k,s)} \triangleq L_{\pi_1(k,s), (k,s)}^{\mathbf{\Pi}_1}  + \sum_{j=1}^{K}{ \sum_{t=1}^{S}{ \left| h_{\pi_1(k,s) \pi_1(j,t)}^{(r)} \right|^2 } },
\end{equation}
where $k \in \{ 1, 2, \cdots, K \}$ and $s \in \{ 1, 2, \cdots, S \}$.
For the first hop (corresponding to the odd time slot $l_o$), the received $\textsf{sinr}_{\pi_1(k,s)}^{(1)}$ at $\mathcal{R}_{\pi_1(k,s)}^{(1)}$ is lower-bounded by
\begin{equation} \label{eq_sinr1_bound}
\begin{split}
\textsf{sinr}_{\pi_1(k,s)}^{(1)} &\geq \frac{\textsf{snr} \left| \mathbf{v}_k^{(s)T} \mathbf{h}_{\pi_1(k,s) k}^{(1)} \right|^2}{1 + \textsf{snr} \left( L_{\!\pi\!_1\!(k,s),\!(k,s)}^{\mathbf{\Pi}_1} + \sum\limits_{j=1}^{K}{ \sum\limits_{t=1}^{S}{ \left| h_{\!\pi_1\!(k,s)\pi\!_2\!(j,t)}^{(r)} \right|^2 } } \right) }\\
&= \frac{\textsf{snr} \left| \mathbf{v}_k^{(s)T} \mathbf{h}_{\pi_1(k,s) k}^{(1)} \right|^2}{1 + \textsf{snr} \tilde{L}_{\pi_1(k,s), (k,s)} }\\
&\geq \frac{\textsf{snr} \left| \mathbf{v}_k^{(s)T} \mathbf{h}_{\pi_1(k,s) k}^{(1)} \right|^2}{1 + \textsf{snr} \sum_{i=1}^{K}{\sum_{t=1}^{S}{ \tilde{L}_{\pi_1(i,t), (i,t)} }}},
\end{split}
\end{equation}
where the first inequality follows from (\ref{eq_metric_set1}) and (\ref{eq_def_tilde_L1}).
For the second hop (corresponding to the even time slot $l_e$), the effective $\textsf{sinr}_{k, \pi_1(k,s)}^{(2)}$ is lower-bounded by
\begin{equation} \label{eq_sinr2_bound}
\begin{split}
\textsf{sinr}_{k, \pi_1(k,s)}^{(2)} &= \frac{\textsf{snr} / \|\mathbf{f}_{k,s}\|^2}{1 + \frac{\textsf{snr} }{\|\mathbf{f}_{k,s}\|^2} \sum\limits_{j=1 \atop j \neq k}^{K}{ \sum\limits_{t=1}^{S}{ \left\| \mathbf{f}_{k,s}^{H} \mathbf{U}^H_k \mathbf{h}_{k \pi_1(j,t)}^{(2)} \right\|^2 } } }\\
&\geq \frac{\textsf{snr} / \|\mathbf{f}_{k,s}\|^2}{1 + \textsf{snr} \sum_{j=1 \atop j \neq k}^{K}{ \sum_{t=1}^{S}{ \left\| \mathbf{U}^H_k \mathbf{h}_{k \pi_1(j,t)}^{(2)} \right\|^2 } } }\\
&\geq \frac{\textsf{snr} / \|\mathbf{f}_{k,s}\|^2}{1 + \textsf{snr} \sum_{t=1}^{S}{ L_{\pi_1(k,t), (k,t)}^{\mathbf{\Pi}_1} } }\\
&\geq \frac{\textsf{snr} / \|\mathbf{f}_{k,s}\|^2}{1 + \textsf{snr} \sum_{i=1}^{K}{ \sum_{t=1}^{S}{ L_{\pi_1(i,t), (i,t)}^{\mathbf{\Pi}_1} }}}\\
&\geq \frac{\textsf{snr} / \|\mathbf{f}_{k,s}\|^2}{1 + \textsf{snr} \sum_{i=1}^{K}{ \sum_{t=1}^{S}{ \tilde{L}_{\pi_1(i,t), (i,t)} }}},
\end{split}
\end{equation}
where the first inequality holds due to the Cauchy-Schwarz inequality.
Here, the term $\sum_{i=1}^{K}{ \sum_{t=1}^{S}{ \tilde{L}_{\pi_1(i,t), (i,t)} }}$ in (\ref{eq_sinr1_bound}) and (\ref{eq_sinr2_bound}) needs to scale as $\textsf{snr}^{-1}$, i.e.,
\begin{equation} \label{ } \nonumber
\sum_{i=1}^{K}{ \sum_{t=1}^{S}{ \tilde{L}_{\pi_1(i,t), (i,t)} }} = O(\textsf{snr}^{-1}),
\end{equation}
so that both $\textsf{sinr}_{\pi_1(k,s)}^{(1)}$ and $\textsf{sinr}_{k, \pi_1(k,s)}^{(2)}$ scale as $\Omega(\textsf{snr})$ with increasing $\textsf{snr}$, which eventually enables our MS-OND protocol to achieve the DoF of $\frac{L-1}{L}$ per stream from (\ref{eq_DoF_bound1}).
Even if such a bounding technique in (\ref{eq_sinr1_bound}) and (\ref{eq_sinr2_bound}) leads to a loose lower bound on the SINR, it is sufficient to prove our achievability result in terms of DoF and relay scaling laws.

Now, let us turn to the second relay set $\mathbf{\Pi}_2$.
Similarly as in (\ref{eq_sinr1_bound}), for the first hop (corresponding to the even time slot $l_e$), the received $\textsf{sinr}_{\pi_2(k,s)}^{(1)}$ at $\mathcal{R}_{\pi_2(k,s)}^{(1)}$ is lower-bounded by
\begin{equation} \label{ } \nonumber
\textsf{sinr}_{\pi_2(k,s)}^{(1)} \geq \frac{\textsf{snr} \left| \mathbf{v}_k^{(s)T} \mathbf{h}_{\pi_2(k,s) k}^{(1)} \right|^2}{1 + \textsf{snr} \sum_{i=1}^{K}{\sum_{t=1}^{S}{ L_{\pi_2(i,t), (i,t)}^{\mathbf{\Pi}_2} }}},
\end{equation}
where $L_{\pi_2(i,t), (i,t)}^{\mathbf{\Pi}_2}$ indicates the TIL in (\ref{eq_metric_set2}) when $\mathcal{R}_{\pi_2(i,t)}^{(1)}$ is assumed to serve the $t$th data stream of the $i$th S--D pair.
For the second hop (corresponding to the odd time slot $l_o$), the effective $\textsf{sinr}_{k, \pi_1(k,s)}^{(2)}$ is lower-bounded by
\begin{equation} \label{ } \nonumber
\textsf{sinr}_{k, \pi_2(k,s)}^{(2)} \geq \frac{\textsf{snr} / \|\mathbf{f}_{k,s}\|^2}{1 + \textsf{snr} \sum_{i=1}^{K}{ \sum_{t=1}^{S}{ L^{\mathbf{\Pi}_2}_{\pi_2(i,t), (i,t)} }}}.
\end{equation}

The next step is thus to characterize the three metrics $L_{j, (k,s)}^{\mathbf{\Pi}_1}$, $L_{j, (k,s)}^{\mathbf{\Pi}_2}$, and $\tilde{L}_{j, (k,s)}$ by computing their cumulative distribution functions (CDFs), where $j \in \{1,2, \cdots, N\}$, $k \in \{1, 2, \cdots, K\}$, and $s \in \{1, 2, \cdots, S\}$.
The three metrics are used to analyze a lower bound on the DoF and the required relay scaling law in the model under consideration.
Since it is obvious to show that the CDF of $\tilde{L}_{j, (k,s)}$ is identical to that of $L_{j, (k,s)}^{\mathbf{\Pi}_2}$, we focus only on the characterization of $L_{j, (k,s)}^{\mathbf{\Pi}_2}$.
The scheduling metric $L_{j, (k,s)}^{\mathbf{\Pi}_1}$ follows the chi-square distribution with $2(2SK - S - 1)$ degrees of freedom since it represents the sum of i.i.d. $2SK - S - 1$ chi-square random variables with $2$ degrees of freedom.
Similarly, the TIL $L_{j, (k,s)}^{\mathbf{\Pi}_2}$ follows the chi-square distribution with $2(3SK - S - 1)$ degrees of freedom.
The CDFs of the two variables $L_{n, (k,s)}^{\mathbf{\Pi}_1}$ and $L_{n, (k,s)}^{\mathbf{\Pi}_2}$ are thus given by
\begin{equation}\label{ } \nonumber
\mathcal{F}_{L1}(l) = \frac{\gamma(2SK - S - 1, l/2)}{\Gamma(2SK - S - 1)}
\end{equation}
and
\begin{equation} \label{ } \nonumber
\mathcal{F}_{L2}(l) = \frac{\gamma(3SK - S - 1, l/2)}{\Gamma(3SK - S - 1)},
\end{equation}
respectively.
Here, $\Gamma(z) = \int_{0}^{\infty} t^{z-1} e^{-t} dt$ is the Gamma function and $\gamma(z, x) = \int_{0}^{x} t^{z-1} e^{-t} dt$ is the lower incomplete Gamma function \cite[eqn. (8.310.1)]{gradshteyn2006table}.
For further analytical tractability, we introduce the following lemma.
\begin{lemma} \label{lemma_CDF}
For any $0 < l \leq 2$, the CDFs of the random variables $L_{n, (k,s)}^{\mathbf{\Pi}_1}$ and $L_{n, (k,s)}^{\mathbf{\Pi}_2}$ are lower-bounded by 
\begin{equation}\label{ } \nonumber
\mathcal{F}_{L1}(l) \geq C_1 l^{2SK-S-1}
\end{equation}
and
\begin{equation}\label{ } \nonumber
\mathcal{F}_{L2}(l) \geq C_2 l^{3SK-S-1},
\end{equation}
respectively, where
\begin{equation} \label{ } \nonumber
C_1 = \frac{e^{-1} 2^{2SK-S-1}}{\Gamma(2SK - S - 1)}
\end{equation}
and
\begin{equation} \label{eq_C2}
C_2 = \frac{e^{-1} 2^{3SK-S-1}}{\Gamma(3SK - S - 1)}.
\end{equation}
\end{lemma}
\begin{proof}
The detailed proof is omitted here since it essentially follows the similar line to the proof of \cite[Lemma 1]{jung2012opportunistic} with a slight modification.
\end{proof}

In the following theorem, we establish our first main result by deriving a lower bound on $\textsf{DoF}_\text{total}$ in the multi-antenna $K \times N \times K$ channel with interfering relay nodes.
\begin{theorem} \label{theorem_scaling}
Suppose that the MS-OND protocol with alternate relaying is used for the $K \times N \times K$ channel with interfering relay nodes when source and destination nodes are equipped with $M$ antennas and each source node transmits $S$ independent data streams.
Then, for $L$ data transmission time slots,
\begin{equation} \label{ } \nonumber %\nonumber
\textsf{DoF}_\text{total} \geq \frac{L-1}{L} SK
\end{equation}
is achievable if $N = \omega \left( \textsf{snr}^{3SK - S - 1} \right)$.
\end{theorem}

\begin{proof}[Proof]
A lower bound on the achievable $\textsf{DoF}_\text{total}$ is given by $\textsf{DoF}_\text{total} \geq \mathcal{P}_o \frac{L-1}{L} SK$, which reveals that $\frac{L-1}{L} SK$ DoF is achievable for a fraction $\mathcal{P}_o$ of the time for actual transmission, where
\begin{equation} \label{eq_Po_definition}
\begin{split}
\mathcal{P}_o = &\lim_{\textsf{snr} \rightarrow \infty} \Pr \left\{ \textsf{snr} \sum_{k=1}^{K}{ \sum_{s=1}^{S}{ \tilde{L}_{\pi_1(k,s),(k,s)} } } \leq \epsilon_0 \right. \\ &\left. \text{ and }\textsf{snr} \sum_{k=1}^{K}{ \sum_{s=1}^{S}{ L_{\pi_2(k,s), (k,s)}^{\mathbf{\Pi_2}} } } \leq \epsilon_0 \right\}
\end{split}
\end{equation}
for a small $\epsilon_0>0$ independent of $\textsf{snr}$.

We now examine the relay scaling condition such that $\mathcal{P}_o$ converges to one with high probability.
For simplicity, suppose that $\mathbf{\Pi_1}$ and $\mathbf{\Pi_2}$ are selected out of two mutually exclusive relaying candidate sets $\mathcal{N}_1$ and $\mathcal{N}_2$, respectively, i.e., $\mathcal{N}_1, \mathcal{N}_2 \subset \{1, 2, \cdots, N\}, \mathcal{N}_1\cap \mathcal{N}_2 = \emptyset, \mathcal{N}_1\cup \mathcal{N}_2 = \{1, 2, \cdots, N\}, \mathbf{\Pi_1} \subset \mathcal{N}_1, \text{and } \mathbf{\Pi_2} \subset \mathcal{N}_2$.
Then, we are interested in examing how $|\mathcal{N}_1|$ and $|\mathcal{N}_2|$ scale with $\textsf{snr}$ in order to guarantee that $\mathcal{P}_o$ tends to approach one, where $|\mathcal{N}_b|$ denotes the cardinality of $\mathcal{N}_b$ for $b = \{1,2\}$.
From (\ref{eq_Po_definition}), we further have
\begin{equation} \label{eq_Po_simplified}
\begin{split}
\mathcal{P}_o = &\lim_{\textsf{snr} \rightarrow \infty} \left( \Pr \left\{ \textsf{snr} \sum_{k=1}^{K}{ \sum_{s=1}^{S}{ \tilde{L}_{\pi_1(k,s),(k,s)} } } \leq \epsilon_0 \right\} \right.
\\ &\left. \cdot \Pr \left\{ \textsf{snr} \sum_{k=1}^{K}{ \sum_{s=1}^{S}{ L_{\pi_2(k,s), (k,s)}^{\mathbf{\Pi_2}} } } \leq \epsilon_0 \right\}  \right).
\end{split}
\end{equation}

Let $\mathcal{B}_{k,s}$ denote the set of the remaining relay candidates in $\mathcal{N}_2$ after the relay $\mathcal{R}_{\pi_2(k,s)}$ has been selected to deliver the $s$th stream of the $k$th S--D pair (note that $\mathcal{R}_{\pi_2(k,s)}$ belongs to $\mathbf{\Pi_2}$), where the cardinality of $\mathcal{B}_{k,s}$ is denoted by $|\mathcal{B}_{k,s}|$.
For a constant $\epsilon_0>0$ independent of $\textsf{snr}$, the second term in (\ref{eq_Po_simplified}) can be lower-bounded by
\begin{equation}
\label{eq_TIL2_bound1}
\begin{split}
&\Pr \left \{ \sum_{k=1}^{K}{ \sum_{s=1}^{S}{ L_{\pi_2(k,s),(k,s)}^{\mathbf{\Pi_2}} } } \leq \frac{\epsilon_0}{\textsf{snr}} \right\}\\
&\geq 1 - \Pr \left \{ \max_{1 \leq k \leq K \atop 1 \leq s \leq S}{ L_{\pi_2(k,s),(k,s)}^{\mathbf{\Pi_2}} } \leq \frac{\epsilon_0}{SK \textsf{snr}} \right\}\\
&\geq 1 - \Pr \left \{ \exists k,s: L_{\pi_2(k,s),(k,s)}^{\mathbf{\Pi_2}} \leq \frac{\epsilon_0}{SK \textsf{snr}} \right\}\\
&\geq 1 - \sum_{k=1}^{K}{ \sum_{s=1}^{S}{}} \Pr \left \{ L_{\pi_2(k,s),(k,s)}^{\mathbf{\Pi_2}} \leq \frac{\epsilon_0}{SK \textsf{snr}} \right\}\\
&\geq 1 - SK \Pr \left \{ \max_{j \in \mathcal{N}_2}{ L_{j,(k,s)}^{\mathbf{\Pi_2}} } \leq \frac{\epsilon_0}{SK \textsf{snr}} \right\}\\
&\geq 1 - SK \left( 1 - \mathcal{F}_L \left( \frac{\epsilon_0}{SK \textsf{snr}} \right) \right)^{| \mathcal{B}_{k,s} |}\\
&\geq 1 - SK \left( 1 - C_2 \left( \frac{\epsilon_0}{SK \textsf{snr}} \right)^{3SK - S - 1} \right)^{| \mathcal{N}_2 | - SK + 1}.
\end{split}
\end{equation}
We now turn to the first term in (\ref{eq_Po_simplified}), which can be lower-bounded by
\begin{equation} \label{eq_TIL1_bound0}
\begin{split}
&\Pr \left \{ \sum_{k=1}^{K}{ \sum_{s=1}^{S}{ \tilde{L}_{\pi_1(k,s),(k,s)} } } \leq \frac{\epsilon_0}{\textsf{snr}} \right\}\\
&\geq 1 - \Pr \left \{ \max_{1 \leq k \leq K \atop 1 \leq s \leq S}{ \tilde{L}_{\pi_1(k,s),(k,s)} } \leq \frac{\epsilon_0}{SK \textsf{snr}} \right\}\\
&= \left( \Pr \left \{ \tilde{L}_{\pi_1(k,s),(k,s)} \leq \frac{\epsilon_0}{SK \textsf{snr}} \right\} \right)^{SK},
\end{split}
\end{equation}
where the equality holds due to the fact that $\tilde{L}_{\pi_1(k,s),(k,s)}$ and $\tilde{L}_{\pi_1(j,t),(j,t)}$ for $i \neq j$ or $s \neq t$ are functions of different random variables and thus are independent of each other.
Let $K_{k,s} = \sum_{j=1}^{K}{ \sum_{t=1}^{S}{\left| h_{\pi_1(k,s)\pi_2(j,t)}^{(r)} \right|^2} }$.
From (\ref{eq_sinr1_bound}), we then have
\begin{equation}
\label{eq_TIL1_bound1}
\begin{split}
&\Pr \left \{ \tilde{L}_{\pi_1(k,s),(k,s)} \leq \frac{\epsilon_0}{SK \textsf{snr}} \right\}\\
&= 1 - \Pr \left \{ L_{\pi_1(k,s),(k,s)}^{\mathbf{\Pi_1}} + K_{k,s} \geq \frac{\epsilon_0}{SK \textsf{snr}} \right\}\\
&\geq 1 - \Pr \left \{ L_{\pi_1(k,s),(k,s)}^{\mathbf{\Pi_1}} \geq \frac{\epsilon_0}{2SK \textsf{snr}} \right\}\\
&\quad \ - \Pr \left \{ K_{k,s} \geq \frac{\epsilon_0}{2SK \textsf{snr}} 
\right\}.
\end{split}
\end{equation}
In the same manner, let $\mathcal{A}_{k,s}$ denote the set of the remaining relay candidates in $\mathcal{N}_1$ after the relay $\mathcal{R}_{\pi_1(k,s)}$ has been selected to deliver the $s$th stream of the $k$th S--D pair (note that $\mathcal{R}_{\pi_1(k,s)}$ belongs to $\mathbf{\Pi_1}$), and let $|\mathcal{A}_{k,s}|$ denote the cardinality of $\mathcal{A}_{k,s}$.
Then, it follows that
\begin{equation}
\label{eq_TIL1_bound2}
\begin{split}
&1 - \Pr \left \{ L_{\pi_1(k,s),(k,s)}^{\mathbf{\Pi_1}} \geq \frac{\epsilon_0}{2SK \textsf{snr}} \right\}\\
&= 1 - \Pr \left \{ \min_{j \in \mathcal{N}_1}{ L_{j,(k,s)}^{\mathbf{\Pi_1}} } \geq \frac{\epsilon_0}{2SK \textsf{snr}} \right\}\\
&= 1 - \left( 1 - \mathcal{F}_{\tilde{L}} \left( \frac{\epsilon_0}{2SK \textsf{snr}} \right) \right)^{|\mathcal{A}_{k,s}|}\\
&\geq 1 - \left( 1 - C_1 \left( \frac{\epsilon_0}{2SK \textsf{snr}} \right)^{2SK - S -1} \right)^{|\mathcal{N}_1|-SK+1}.
\end{split}
\end{equation}

From (\ref{eq_TIL2_bound1}), (\ref{eq_TIL1_bound1}), and (\ref{eq_TIL1_bound2}), it is obvious that if $|\mathcal{N}_1|$ and $|\mathcal{N}_2|$ scale faster than $\textsf{snr}^{2SK - S - 1}$ and $\textsf{snr}^{3SK - S - 1}$, respectively, then
\begin{equation} \label{ } \nonumber
\lim_{\textsf{snr} \rightarrow \infty}{\left( 1 - C_1 \left( \frac{\epsilon_0}{2SK \textsf{snr}} \right)^{2SK - S -1} \right)^{|\mathcal{N}_1|-SK+1} } = 0
\end{equation}

\begin{equation} \label{ } \nonumber
\lim_{\textsf{snr} \rightarrow \infty}{\left( 1 - C_2 \left( \frac{\epsilon_0}{SK \textsf{snr}} \right)^{3SK - S - 1} \right)^{| \mathcal{N}_2 | - SK + 1} } = 0.
\end{equation}

Under this condition, $\mathcal{P}_o$ asymptotically approaches one, which means that DoF of $\frac{L-1}{L} SK$ is achievable with high probability if $N = |\mathcal{N}_1| + |\mathcal{N}_2| = \omega \left( \textsf{snr}^{3SK - S - 1} \right)$.
This completes the proof.
\end{proof}

%comprehensive meaning of this theorem, including for large L, we have optimal SK DoF.
Note that the achievable bound on $\textsf{DoF}_\text{total}$ asymptotically approaches $SK$ for large $L$, which implies that our system operates in virtual full-duplex mode.
The parameter $N$ required to achieve $SK$ DoF needs to increase exponentially with not only the number of S--D pairs, $K$, but also the number of data streams per S--D pair, $S$, so that the sum of $3SK - S - 1$ interference terms in the scheduling metric TIL in (\ref{eq_metric_set2}) does not scale with increasing $\textsf{snr}$ at each relay node.
Here, from the perspective of each relay node in $\mathbf{\Pi}_2$, the SNR exponent $3SK - S - 1$, indicating the total number of interference links, stems from the following three factors: i) the sum of interference power received from other spatial beams including not only the beams of other source nodes (i.e., $S(K-1)$ interfering links) but also other beams created by the same source node (i.e., $S-1$ interfering links); ii) the sum of interference power generating to other destination nodes (i.e., $S(K-1)$ interfering links); and iii) the sum of inter-relay interference power generated from $\mathbf{\Pi}_1$ (i.e., $SK$ interfering links).
Thus, it is possible to enhance the achievable DoF in our MS-OND framework by increasing $S$ (i.e., sending more data streams per S--D pair), at the cost of increased number of relay nodes.

\emph{Remark 1:} $SK$ DoF can be achieved by using the MS-OND protocol if $N$ scales faster than $\textsf{snr}^{3SK - S - 1}$ and the number of transmission slots in one block, $L$, is sufficiently large.
In this case, all the interference signals are almost nulled out at each selected relay node via RBF for the first hop and are then almost aligned at each destination node via OIA for the second hop by exploiting the multiuser diversity gain.
In other words, by applying the MS-OND protocol to the interference-limited multi-antenna $K \times N \times K$ channel such that the channel links are inherently coupled with each other, the links across each S--D path via $2S$ relay nodes can be completely \emph{decoupled} with each other, thus enabling us to achieve the same DoF as in the interference-free channel case.
We also note that deploying multiple antennas at relay nodes does not further increase the DoF in our system since the achievable DoF cannot go beyond $MK$ and the MS-OND protocol has already achieved such DoF using a single antenna at each relay node. However, the multi-antenna configuration on the relay nodes would relieve the relay scaling condition required to achieve a target DoF owing to more effective interference management with multiple antennas.

\emph{Remark 2:} We also show an upper bound on $\textsf{DoF}_\text{total}$ by using the cut-set bound argument similarly as in the single-antenna OND protocol \cite[Section 4]{shin2017opportunistic}.
Suppose that $\tilde{N}$ relay nodes are active, where $\tilde{N} \in \{1, 2, \cdots, N\}$.
Consider two cuts $L_1$ and $L_2$ dividing our network into two parts in a different manner.
Then, we can create an $MK \times (\tilde{N} + MK)$ MIMO channel under the cut $L_1$.
Similarly, a $(\tilde{N} + MK) \times MK$ MIMO channel is obtained under the cut $L_2$.
The DoF of both the two MIMO channels is thus upper-bounded by $MK$.
It turns out that our lower bound in Theorem \ref{theorem_scaling} matches this upper bound for $S=M$ and large $L$.

\subsection{Decaying Rate of TIL}
In this section, we analyze the decaying rate of the TIL under the MS-OND protocol with alternate relaying, which can provide insights on how the TIL is bounded with increasing $\textsf{snr}$.
It is found that the desired relay scaling law is closely related to the decaying rate of the TIL with respect to $N$ for given $\textsf{snr}$.

Let $L_{SK\text{-th}}$ denote the TIL of the $SK$-th selected relay node in time order (i.e., the last selected relay node).
Thus, $L_{SK\text{-th}}$ is the largest one among all the $SK$ TIL values of the selected relay nodes, which can be used to evaluate the interference controlling performance on the DoF.
By Markov's inequality, a lower bound on the $L_{SK\text{-th}}$'s expectation with respect to $N$ is given by
\begin{equation} \label{ } \nonumber
\mathbb{E}\left[ \frac{1}{L_{SK\text{-th}}} \right] \geq \frac{1}{\epsilon} \Pr(L_{SK\text{-th}} \leq \epsilon),
\end{equation}
where the inequality always holds for $\epsilon > 0$.
Let $\mathcal{P}_{SK}(\epsilon)$ denote the probability that there exist only $SK$ relay nodes among $N$ relay candidates satisfying TIL $\leq \epsilon$, which is expressed as
\begin{equation}
\label{eq_prob_N_SK}
\mathcal{P}_{SK}(\epsilon) = {N \choose SK} {\mathcal{F}_L(\epsilon)}^{SK} (1 - \mathcal{F}_L(\epsilon))^{N-SK},
\end{equation}
where $\mathcal{F}_L(\epsilon)$ is the CDF of the TIL.
Since $\Pr(L_{SK\text{-th}} \leq \epsilon)$ is lower-bounded by $\Pr(L_{SK\text{-th}} \leq \epsilon) \geq \mathcal{P}_{SK}(\epsilon)$, a lower bound on the average decaying rate of the TIL is given by
\begin{equation} \label{eq_TIL_LB}
\mathbb{E}\left[ \frac{1}{L_{SK\text{-th}}} \right] \geq \frac{1}{\epsilon} \mathcal{P}_{SK}(\epsilon).
\end{equation}

The next step is to find the parameter $\hat{\epsilon}$ that maximizes $\mathcal{P}_{SK}(\hat{\epsilon})$ with respect to $\epsilon$ in order to provide the tightest lower bound.
\begin{lemma} \label{lemma_epsilon}
When a constant $\hat{\epsilon}$ satisfies the condition $\mathcal{F}_L(\hat{\epsilon}) =  \frac{SK}{N}$, $\mathcal{P}_{SK}(\hat{\epsilon})$ in (\ref{eq_prob_N_SK}) is maximized for a given $N$.
\end{lemma}
\begin{proof}
We refer to \cite[Lemma 2]{shin2017opportunistic} for the proof.
\end{proof}

Now, we establish our second main theorem, which shows a lower bound on the decaying rate of the TIL with respect to $N$.

\begin{theorem} \label{theorem_TIL_decaying_rate}
Suppose that the MS-OND protocol with alternate relaying is used for the $K \times N \times K$ channel with interfering relay nodes when source and destination nodes are equipped with $M$ antennas and each source node transmits $S$ independent data streams.
Then, the decaying rate of the TIL is lower-bounded by
\begin{equation}
%\nonumber
\label{ } \nonumber
\mathbb{E}\left[ \frac{1}{L_{SK\text{-th}}} \right] \geq \Theta \left( N^{\frac{1}{3SK - S -1}} \right).
\end{equation}
\end{theorem}
\begin{proof}[Proof]
As shown in (\ref{eq_TIL_LB}), the decaying rate of the TIL is lower-bounded by the maximum of $\frac{1}{\epsilon} \mathcal{P}_{SK}(\epsilon)$ over $\epsilon$.
By Lemma \ref{lemma_epsilon}, $\mathcal{P}_{SK}(\hat{\epsilon})$ is maximized when $\hat{\epsilon} = \mathcal{F}_L^{-1} \left( \frac{SK}{N} \right)$.
Thus, we have
\begin{equation} \label{ } \nonumber
\begin{split}
\mathbb{E}\left[ \frac{1}{L_{SK\text{-th}}} \right] &\geq \frac{1}{ \mathcal{F}_L^{-1} \left(\!\frac{SK}{N}\!\right) } {N \choose SK} \left(\!\frac{SK}{N} \right)^{\!SK}\!\left(1 - \frac{SK}{N}\!\right)^{\!N-SK}\\
& \geq \frac{1}{ \mathcal{F}_L^{-1}\!\left(\!\frac{SK}{N}\!\right) } \left(\!\frac{N-SK+1}{N}\!\right)^{\!SK}\!\left(\!1 - \frac{SK}{N}\!\right)^{\!N-SK}\\
& \geq \frac{1}{ \mathcal{F}_L^{-1} \left( \frac{SK}{N} \right) } \left( \frac{1}{SK} \right)^{SK} e^{-SK}\\
& \geq \Theta \left( N^{\frac{1}{3SK - S -1}} \right),
\end{split}
\end{equation}
where the second and third inequalities hold since ${N \choose SK} \geq \left( 1 - \frac{SK}{N} \right)^{SK}$ and $\left( 1 - \frac{SK}{N} \right)^{N-SK} \leq \left( 1 - \frac{SK}{N} \right)^N \leq e^{-SK}$, respectively.
By Lemma \ref{lemma_CDF}, it follows that $\frac{1}{\mathcal{F}_L^{-1} \left( \frac{SK}{N} \right)} \geq \left( \frac{C_2 N}{SK} \right)^{\frac{1}{3SK-S-1}}$, where $C_2$ is given by (\ref{eq_C2}).
Hence, the last inequality also holds, which completes the proof.
\end{proof}

\subsection{MS-OND Without Alternate Relaying}
For comparison, the MS-OND protocol \emph{without} alternate relaying is also described in this section.
In the protocol, only the first relay set $\mathbf{\Pi}_1$ participates in data reception and forwarding.
In other words, the second relay set $\mathbf{\Pi}_2$ does not need to be selected for the MS-OND protocol without alternate relaying.
The overall procedure during one transmission block is described as follows.

\textbf{(a) Odd time slot $l_o \in \{1, 3, \cdots, L-1\}$}: The $K$ source nodes transmit their encoded symbols to the $SK$ selected relay nodes in $\mathbf{\Pi}_1$.
For instance, the source $\mathcal{S}_k$ transmits symbols $x_{k,1}^{(1)}{(l_o)}, \cdots, x_{k,S}^{(1)}{(l_o)}$ along with its $S$ randomly generated spatial beams, where $k \in \{1, 2, \cdots, K\}$.
The relay $\mathcal{R}_{\pi_1(k, s)}$ receives the desired symbol $x_{k, s}^{(1)}{(l_o)}$ on the $s$th beam of $\mathcal{S}_k$, where $s \in \{1, 2, \cdots, S\}$.

\textbf{(b) Even time slot $l_e \in \{2, 4, \cdots, L\}$}:
The selected relay nodes in $\mathbf{\Pi}_1$ forward their re-encoded symbols to the intended destination nodes.
For instance, the relay nodes $\mathcal{R}_{\pi_1(k,1)}, \cdots, \mathcal{R}_{\pi_1(k,S)}$ transmit their symbols $x_{\pi_1(k,1)}^{(2)}{(l_e-1)}, \allowbreak \cdots, \allowbreak x_{\pi_1(k,S)}^{(2)}{(l_e-1)}$ to the destination $\mathcal{D}_k$ while the interfering signals to the other destination nodes are opportunistically aligned to their interference spaces.
If the relay nodes in $\mathbf{\Pi}_1$ successfully decode the corresponding symbols, then $x_{\pi_1(k, s)}^{(2)}{(l_e-1)}$ would be the same as $x_{k, s}^{(1)}{(l_e-1)}$.

When the relay $\mathcal{R}_n$ is assumed to serve the $s$th beam of the $k$th S--D pair for $n \in \{1, 2, \cdots,N\}$, $s \in \{1, 2, \cdots,S\}$, and $k \in \{1, 2, \cdots,K\}$, it computes the scheduling metric $L_{n, (k,s)}^{\mathbf{\Pi}_1}$ in (\ref{eq_metric_set1}).
With the computed $L_{n, (k,s)}^{\mathbf{\Pi}_1}$, a timer based method is used for relay selection similarly as in Section \ref{subsec_timer}.
The DoF achieved by the MS-OND protocol without alternate relaying is shown in the following theorem.
\begin{theorem} \label{theorem_scaling_hd}
Suppose that the MS-OND protocol without alternate relaying is used for the $K \times N \times K$ channel with interfering relay nodes when source and destination nodes are equipped with $M$ antennas and each source node transmits $S$ independent data streams.
Then,
\begin{equation} \label{ } \nonumber
\textsf{DoF}_\text{total} \geq \frac{SK}{2}
\end{equation}
is achievable if $N = \omega \left( \textsf{snr}^{2SK - S - 1} \right)$.
\end{theorem}

\begin{proof}[Proof]
The detailed proof of this argument is omitted here since it basically follows the same line as the proof of Theorem \ref{theorem_scaling}.
\end{proof}

\emph{Remark 3:} On the one hand, it is observed from Theorem \ref{theorem_scaling_hd} that for \emph{given $S$} (i.e., the fixed number of data streams per S--D pair), half of $SK$ DoF can be achieved by the MS-OND protocol without alternate relaying under a less stringent relay scaling condition compared to the result in Theorem \ref{theorem_scaling}.
For the MS-OND protocol with alternate relaying, in order to achieve the DoF of $\frac{L-1}{L}SK$, at least $\textsf{snr}^{3SK - S - 1}$ relay nodes are required.
For the MS-OND protocol without alternate relaying, in order to achieve the DoF of $\frac{SK}{2}$, at least $\textsf{snr}^{2SK - S - 1}$ relay nodes are required.
For instance, when $S=1$, $K=2$, and $\textsf{snr} = 5$ (in linear scale), $5^{4} = 625$ relay nodes are necessary to achieve the DoF of almost 2 along with the MS-OND protocol with alternate relaying.
Otherwise (i.e., when the number of relay nodes is less than the required number in practice), no DoF is guaranteed due to the inherent limitation of the opportunistic transmission mechanism.
On the other hand, to achieve a \emph{fixed target DoF}, the MS-OND protocol without alternate relaying requires more relay nodes when $K \geq 2$.
For instance, suppose that the target DoF is $2K$.
Then, the relay scaling condition required under the MS-OND protocol with alternate relaying is $\omega \left( \textsf{snr}^{6K -3} \right)$, which is less stringent than another condition $\omega \left( \textsf{snr}^{8K -5} \right)$ required under the MS-OND protocol without alternate relaying.
This comes from the fact that to achieve $2K$ DoF, the MS-OND protocol without alternate relaying uses twice as many data streams.
Moreover, it is seen that in a \emph{finite $N$ regime}, there exist \textsf{snr} regimes where the MS-OND protocol without alternate relaying outperforms that with alternate relaying in terms of sum-rates, which will be numerically verified in Section \ref{sec_results}.

%\newpage
\section{Numerical Evaluation} \label{sec_results}
In this section, we perform computer simulations to validate our analytical results for finite $N$ and $\textsf{snr}$.
In our simulations, the channel coefficients in (\ref{eq_rx_signal_hop1}) and (\ref{eq_rx_signal_hop2}) are generated $1 \times 10^5$ times for each system configuration.

\begin{figure}
\begin {center}
\epsfig{file=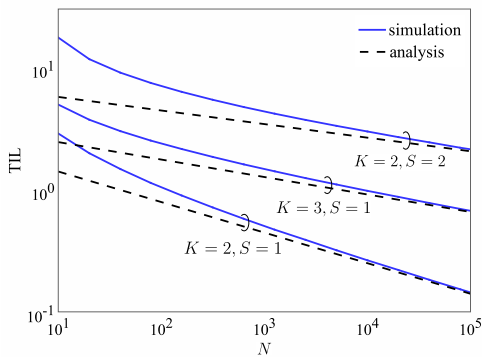, width=0.73\hsize}
\end {center}
\caption{The average TIL versus $N$ for various $K$ and $S$ when the MS-OND protocol with alternate relaying is used in the multi-antenna $K\times N\times K$ channel with interfering relay nodes, where $M = 4$.}
\label{fig_msond_til}
\end{figure}
In Fig. \ref{fig_msond_til}, when the MS-OND protocol with alternate relaying is employed in the multi-antenna $K\times N\times K$ channel with interfering relay nodes, a log-log plot of the average TIL versus $N$ is shown according to various parameter settings including $(K,S)=(2,1)$, $(3,1)$, and $(2,2)$.
The number of antennas at each S--D pair is set to $M=4$ for all the simulations.
It is numerically found that the TIL tends to linearly decrease with $N$ for large $N$.
It is further seen how many relay nodes are required to guarantee that the TIL is less than a certain small constant for given parameters $K$ and $S$.
In this figure, the dotted lines are also plotted from the theoretical result from Theorem \ref{theorem_TIL_decaying_rate} with a proper bias to check the slope of the TIL.
We can see that the decaying rate of the TILs is consistent with the relay scaling law condition in Theorem \ref{theorem_scaling}.
More specifically, the TIL is reduced as $N$ increases with slopes of $\frac{1}{4}$ for $(K,S)=(2,1)$, $\frac{1}{7}$ for $(K,S)=(3,1)$, and $\frac{1}{9}$ for $(K,S)=(2,2)$, respectively.

   \begin{figure}
   \begin {center}
     \subfloat[$S = 1$.]{%
       \includegraphics[width=0.36\textwidth]{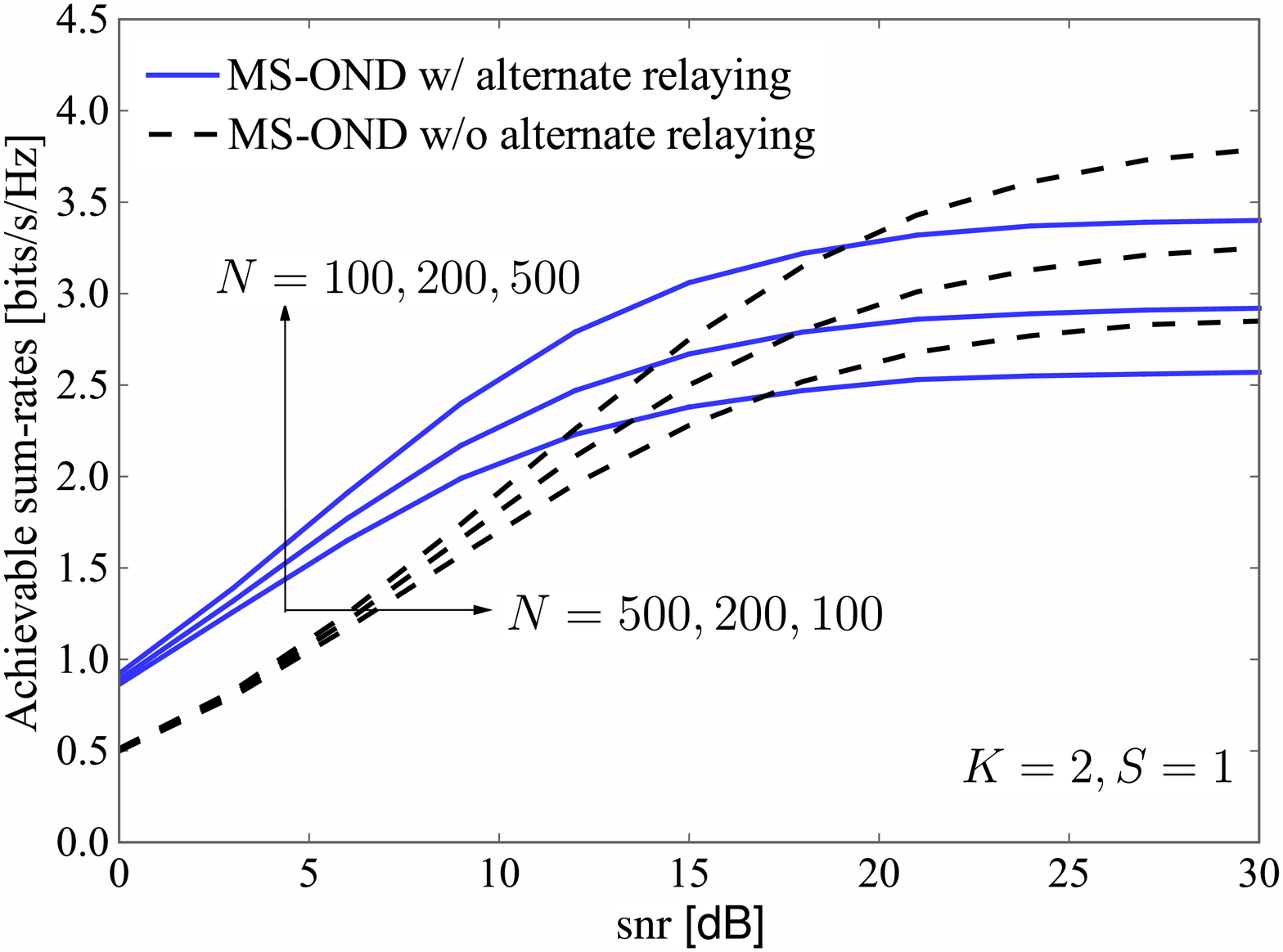}
     }
     \hfill
     \subfloat[$S = 2$.]{%
       \includegraphics[width=0.36\textwidth]{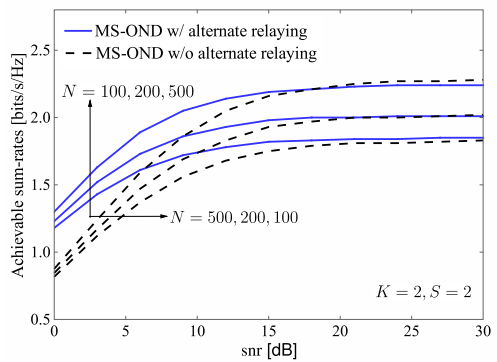}
     }
     \end {center}
     \caption{The achievable sum-rates versus \textsf{snr}, where $K = 2$, $M=4$, and $N \in \{100,200,500\}$ in the multi-antenna $K\times N\times K$ channel with interfering relay nodes. The MS-OND protocols with and without alternate relaying are compared.}
     \label{fig_msond_rate_n_k2s1}
   \end{figure}
Figure \ref{fig_msond_rate_n_k2s1} illustrates the sum-rates achieved by the MS-OND protocols with and without alternate relaying in the $K \times N \times K$ channel with inter-relay interference versus \textsf{snr} (in dB scale), where $(K,S)=(2,1)$ and $(2,2)$, $N = \{100, 200, 500\}$, and $M=4$.
It is seen that in a finite $N$ regime, there exists the case where the MS-OND protocol without alternate relaying outperforms the MS-OND protocol with alternate relaying.
This is because for finite $N$, the achievable sum-rates for the alternate relaying case tend to approach a floor in a low or moderate SNR regime due to more residual interference in each dimension.
The sum-rates for the MS-OND protocol without alternate relaying increase faster with \textsf{snr} in the low or moderate \textsf{snr} regime owing to less residual interference, compared to the MS-OND protocol with alternate relaying.
However, the sum-rates achieved by both protocols tend to get saturated in the high \textsf{snr} regime because of more stringent relaying scaling condition for larger $S$  (refer to Theorems \ref{theorem_scaling} and \ref{theorem_scaling_hd}).
These observations motivate us to operate our system in switch mode where the relaying scheme is chosen between the MS-OND protocols with and without alternate relaying depending on not only the operating regime but also the system configuration.

\begin{figure}
\begin {center}
\epsfig{file=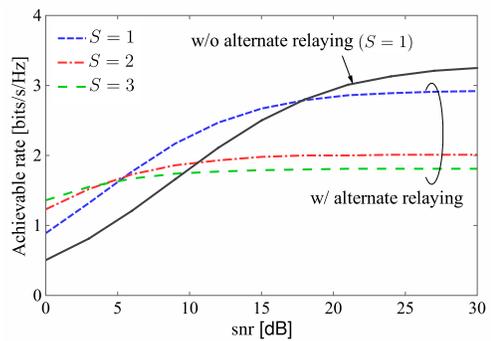, width=0.73\hsize}
\end {center}
\caption{The achievable sum-rates versus \textsf{snr} for various numbers of data streams per S--D pair, $S$, in the multi-antenna $K\times N\times K$ channel with interfering relay nodes, where $K = 2$, $M=4$, and $N=200$. The performance of MS-OND protocol without alternate relaying for $S=1$ is also plotted.}
\label{fig_msond_rate_s_k2n200}
\end{figure}
In Fig. \ref{fig_msond_rate_s_k2n200}, in order to examine which one is dominant between the MS-OND protocols with and without alternate relaying, the sum-rates achieved by both protocols in the $K \times N \times K$ channel versus \textsf{snr} (in dB scale) are plotted according to various $S$ indicating the number of data streams per S--D pair, where $K=2$, $M=4$, and $N=200$.
It is seen that for the case with alternate relaying, large $S$ leads to higher sum-rates in the low \textsf{snr} regime but gets saturated earlier.
Thus, superior performance on the sum-rates can be achieved for small $S$ in the high \textsf{snr} regime.
It is also observed that the MS-OND protocol without alternate relaying for $S=1$ outperforms all other cases with alternate relaying in the very high \textsf{snr} regime since it has the least stringent user scaling condition.

   \begin{figure}
   \begin {center}
     \subfloat[$\textsf{snr} = 15 \text{ [dB]}$.]{%
       \includegraphics[width=0.36\textwidth]{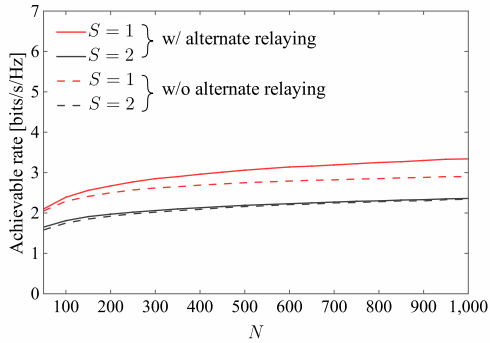}
     }
     \hfill
     \subfloat[$\textsf{snr} = 30 \text{ [dB]}$.]{%
       \includegraphics[width=0.36\textwidth]{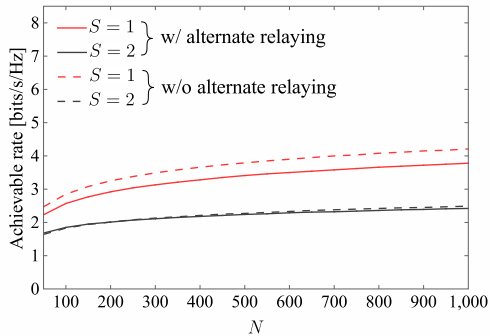}
     }
     \end {center}
     \caption{The achievable sum-rates versus the number of relay nodes $N$ in the multi-antenna $K\times N\times K$ channel with interfering relay nodes, where $K = 2$, $M=4$, and $\textsf{snr} =\{15,30\}$[dB]. The MS-OND protocols with and without alternate relaying are compared.}
     \label{fig_msond_rate_n_snr}
   \end{figure}
Figure \ref{fig_msond_rate_n_snr} illustrates the effect of the number of relay nodes, $N$, on the sum-rate performance for various $S$, where $K = 2$, $M=4$, and $\textsf{snr} =\{15,30\}$ [dB].
Owing to the opportunistic gain, it is obvious that the sum-rate increases with $N$ for all cases.
This observation implies that even if the relay scaling conditions in Theorems \ref{theorem_scaling} and \ref{theorem_scaling_hd} are not fulfilled (in order to guarantee the target DoF), the sum-rate performance can still be enhanced with increasing $N$.
For comparison of the two types of MS-OND protocols, it is found that the MS-OND protocol with alternate relaying outperforms its counterpart (i.e., the one without alternate relaying) when $\textsf{snr} = 15$ [dB], while the MS-OND protocol without alternate relaying achieves higher sum-rates than those of its counterpart when $\textsf{snr} = 30$ [dB].
This is due to the fact that according to Theorems \ref{theorem_scaling} and \ref{theorem_scaling_hd}, more stringent relay scaling condition is required in a higher \textsf{snr} regime, where the MS-OND protocol without alternate relaying relaxes the scaling requirement.

\begin{table*}[!t]
\renewcommand{\arraystretch}{1.2}
  \centering 
  \caption{The lookup table according to operating regimes.} \label{table_lookup}
  \begin{tabular}{| c | c | c | c | c |}
  \hline
  \diagbox{$N$}{Strategy} & \tabincell{c}{I \\ (AR, $S=3$)} & \tabincell{c}{II \\ (AR, $S=2$)} & \tabincell{c}{III \\ (AR, $S=1$)} & \tabincell{c}{IV \\ (NAR, $S=1$)}\\
  \hline
   $50$ & \tabincell{c}{$\textsf{snr} \leq 4$ \\ $T_\text{max} \leq 1.42$} & \tabincell{c}{$4 < \textsf{snr} \leq 5$ \\ $1.42 < T_\text{max} \leq 1.45$} & \tabincell{c}{$5 < \textsf{snr} \leq 15$ \\ $1.45 < T_\text{max} \leq 2.09$} & \tabincell{c}{$15 < \textsf{snr}$ \\ $2.09 < T_\text{max}$}\\
  \hline
   $100$ & \tabincell{c}{$\textsf{snr} \leq 4$ \\ $T_\text{max} \leq 1.52$} & \tabincell{c}{$4 < \textsf{snr} \leq 5$ \\ $1.52 < T_\text{max} \leq 1.57$} & \tabincell{c}{$5 < \textsf{snr} \leq 16$ \\ $1.57 < T_\text{max} \leq 2.43$} & \tabincell{c}{$16 < \textsf{snr}$ \\ $2.43 < T_\text{max}$}\\
  \hline
  $200$ & \tabincell{c}{$\textsf{snr} \leq 3$ \\ $T_\text{max} \leq 1.55$} & \tabincell{c}{$3 < \textsf{snr} \leq 5$ \\ $1.55 < T_\text{max} \leq 1.67$} & \tabincell{c}{$5 < \textsf{snr} \leq 17$ \\ $1.67 < T_\text{max} \leq 2.77$} & \tabincell{c}{$17 < \textsf{snr}$ \\ $2.77 < T_\text{max}$}\\
  \hline
  $500$ & \tabincell{c}{$\textsf{snr} \leq 3$ \\ $T_\text{max} \leq 1.65$} & \tabincell{c}{$3 < \textsf{snr} \leq 5$ \\ $1.65 < T_\text{max} \leq 1.82$} & \tabincell{c}{$5 < \textsf{snr} \leq 19$ \\ $1.82 < T_\text{max} \leq 3.25$} & \tabincell{c}{$19 < \textsf{snr}$ \\ $3.25 < T_\text{max}$}\\
  \hline
  \end{tabular}
  \label{table_delta}
\end{table*}

From the aforementioned observations, to achieve the maximum sum-rate, our system needs to operate in hybrid mode, which switches between the MS-OND protocols with and without alternate relaying and selects proper $S$ depending on the operating regime.
Table \ref{table_lookup} is provided to demonstrate which strategy yields the highest sum-rate for different \textsf{snr} regimes according to various $N$ by selecting one of four strategies I--IV indicated in the table, where we use ``AR'', ``NAR'', and $T_\text{max}$ to denote the MS-OND protocol with alternative relaying, the MS-OND protocol without alternative relaying, and the maximum sum-rate, respectively.
From Table \ref{table_lookup}, the following interesting observations are made for each $N$: the strategy I tends to lead to the highest sum-rate in the very low \textsf{snr} regime; while the strategy IV tends to be dominant in the very high \textsf{snr} regime.

\begin{figure}
\begin {center}
\epsfig{file=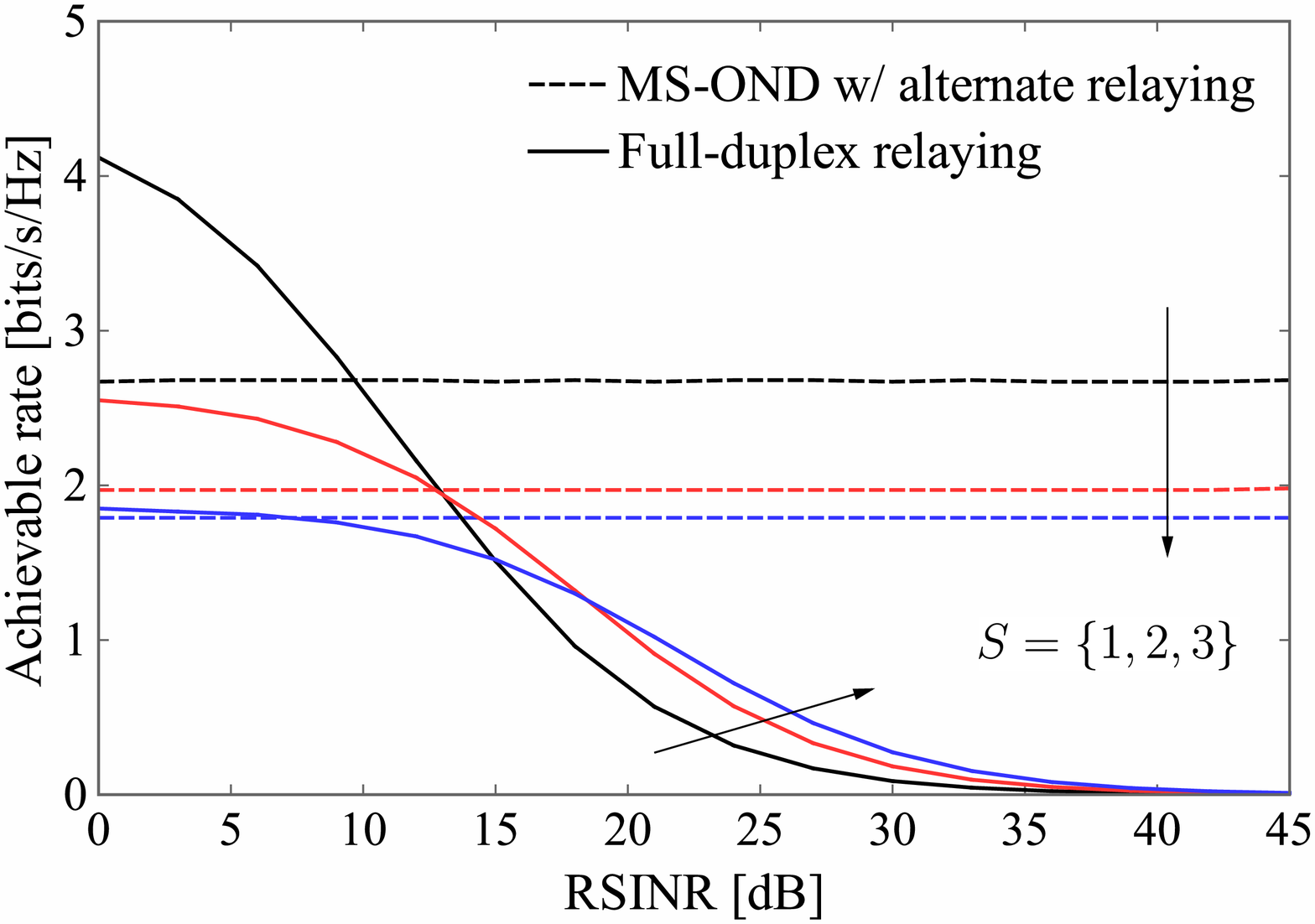, width=0.73\hsize}
\end {center}
\caption{The performance comparison between the full-duplex relaying scheme and the MS-OND protocol with alternate relaying in the multi-antenna $K\times N\times K$ channel with interfering relay nodes, where $\textsf{snr}=15$ [dB], $N=200$, $K=2$, $M = 4$, and $S = \{1, 2, 3\}$. The performance of the MS-OND protocol with alternate relaying and the full-duplex relaying protocol is compared.}
\label{fig_msond_rate_si_snr15}
\end{figure}
Furthermore, it would be meaningful to compare the performance of our protocol with another benchmark scheme in which single-antenna {\em full-duplex} relay nodes are deployed and $SK$ relay nodes are opportunistically selected in the sense of generating or receiving the minimum sum of the interference from other spatial beams during the S--R transmission and the interference leakage to other destination nodes during the R--D transmission.
Unlike our MS-OND protocol, such a full-duplex relaying protocol experiences not only the residual self-interference after SIC but also the full inter-relay interference, whereas it can achieve up to twice as much spectral efficiency as the half-duplex relaying case. Thus, it is not obvious which one is superior to another in the multi-antenna $K\times N\times K$ channel with interfering relay nodes.
In Fig. \ref{fig_msond_rate_si_snr15}, the sum-rates achieved by the MS-OND protocol with alternate relaying and the full-duplex relaying protocol in the
$K\times N\times K$ channel versus residual self-interference-to-noise ratio (RSINR) (in dB scale) are plotted according
to various $S$, where \textsf{snr} = 15 [dB], $N=200$, $K = 2$, and $M = 4$.
It is observed that there exists a crossover between two curves for a given $S$. Specifically, the full-duplex relaying protocol outperforms the MS-OND protocol in a low RSINR regime; but the sum-rates achieved by the full-duplex relaying protocol are significantly reduced with increasing RSINR. This implies that the advantage of full-duplex relaying is guaranteed only when powerful SIC can be implemented at the relay nodes, e.g., an RSINR lower than 10 dB is required for $S=1$. Such a high requirement on the SIC would be quite stringent and challenging under mMTC or IoT networks consisting of low-cost relaying devices.

\section{Conclusion} \label{sec_conclusion}
In this paper, we presented MS-OND to achieve the target DoF in the multi-antenna $K \times N \times K$ channel under a certain relay scaling law, where the source and destination nodes were equipped with $M$ antennas while half-duplex relay nodes are equipped with a single antenna.
The proposed MS-OND protocol that delivers $S$ ($1\le S\le M$) data streams per S--D pair was built upon the conventional OND in the single-antenna setup by leveraging both relay selection and interference management techniques.
Two subsets of relay nodes among $N$ relay candidates were opportunistically selected while using alternate relaying in terms of generating or receiving the minimum TIL.
For interference management, our protocol intelligently integrated RBF for the first hop and OIA for the second hop into the network decoupling framework.
It was shown that our MS-OND protocol asymptotically achieves the optimal $SK$ DoF, provided that the number of relay nodes scales faster than $\textsf{snr}^{3SK - S - 1}$.
Our analytical results were numerically validated through extensive computer simulations.
Moreover, it was provided how the MS-OND protocol works in pratice with a proper transmission and relaying strategy in finite $N$ or $\textsf{snr}$ regimes.
Numerical evaluation showed that the strategy setting large $S$ and adopting alternate relaying provides the best sum-rate performance in the low $\textsf{snr}$ regime; on the contrary, the strategy setting $S = 1$ without alternate relaying outperforms all other cases in the very high $\textsf{snr}$ regime.
Hence, we shed light on the DoF-optimal design of distributed multi-stream transmission protocols based on partial channel knowledge in IoT or mMTC networks with a large number of sensors.

\section*{Acknowledgment}
This research was supported by the Basic Science Research Program through the National Research Foundation of Korea (NRF) funded by the Ministry of Education
(2017R1D1A1A09000835). The material in this paper was presented in part at the IEEE Vehicular Technology Society Asia Pacific Wireless Communications Symposium, Incheon, Republic of Korea, August 2017 \cite{lin2017apwcs}.
Won-Yong Shin is the corresponding author.

\ifCLASSOPTIONcaptionsoff
  \newpage
\fi


\begin{thebibliography}{1}
%\providecommand{\url}[1]{#1}
%\csname url@samestyle\endcsname
%\providecommand{\newblock}{\relax}
%\providecommand{\bibinfo}[2]{#2}
%\providecommand{\BIBentrySTDinterwordspacing}{\spaceskip=0pt\relax}
%\providecommand{\BIBentryALTinterwordstretchfactor}{4}
%\providecommand{\BIBentryALTinterwordspacing}{\spaceskip=\fontdimen2\font plus
%\BIBentryALTinterwordstretchfactor\fontdimen3\font minus
%  \fontdimen4\font\relax}
%\providecommand{\BIBforeignlanguage}[2]{{%
%\expandafter\ifx\csname l@#1\endcsname\relax
%\typeout{** WARNING: IEEEtran.bst: No hyphenation pattern has been}%
%\typeout{** loaded for the language `#1'. Using the pattern for}%
%\typeout{** the default language instead.}%
%\else
%\language=\csname l@#1\endcsname
%\fi
%#2}}
%\providecommand{\BIBdecl}{\relax}
%\BIBdecl

\bibitem{ercan2017rf}
A.~{\"O}. Ercan, O.~Sunay, and I.~F. Akyildiz, ``{RF energy harvesting and
  transfer for spectrum sharing cellular IoT communications in 5G systems},''
  \emph{IEEE Trans. Mobile Comput.}, vol.~17, no.~7, pp. 1680--1694, Jul. 2017.

\bibitem{agiwal2016next}
M.~Agiwal, A.~Roy, and N.~Saxena, ``{Next generation 5G wireless networks: A
  comprehensive survey},'' \emph{IEEE Commun. Surveys Tuts.}, vol.~18, no.~3,
  pp. 1617--1655, 2016.

\bibitem{3gpp2016study}
``Study on scenarios and requirements for next generation access
  technologies,'' \emph{3GPP TR 38.913 V14.3.0}, Jun. 2017.

\bibitem{akpakwu2017survey}
G.~A. Akpakwu, B.~J. Silva, G.~P. Hancke, and A.~M. Abu-Mahfouz, ``A survey on
  {5G} networks for the internet of things: {Communication} technologies and
  challenges,'' \emph{IEEE Access}, vol.~6, pp. 3619--3647, Dec. 2017.

\bibitem{Lee2014a}
N.~Lee and R.~W. {Heath Jr}, ``{Advanced interference management technique:
  Potentials and limitations},'' \emph{IEEE Wireless Commun.}, vol.~23, no.~3,
  pp. 30--38, Jun. 2016.

\bibitem{maddah2008communication}
M.~A. Maddah-Ali, A.~S. Motahari, and A.~K. Khandani, ``{Communication over
  MIMO $X$ channels: Interference alignment, decomposition, and performance
  analysis},'' \emph{IEEE Trans. Inf. Theory}, vol.~54, no.~8, pp. 3457--3470,
  Aug. 2008.

\bibitem{cadambe2008interference}
V.~R. Cadambe and S.~A. Jafar, ``Interference alignment and degrees of freedom
  of the {$K$}-user interference channel,'' \emph{IEEE Trans. Inf. Theory},
  vol.~54, no.~8, pp. 3425--3441, Jul. 2008.

\bibitem{gomadam2011distributed}
K.~Gomadam, V.~R. Cadambe, and S.~A. Jafar, ``A distributed numerical approach
  to interference alignment and applications to wireless interference
  networks,'' \emph{IEEE Trans. Inf. Theory}, vol.~57, no.~6, pp. 3309--3322,
  May 2011.

\bibitem{gou2010degrees}
T.~Gou and S.~A. Jafar, ``{Degrees of freedom of the $K$ user $M \times N$ MIMO
  interference channel},'' \emph{IEEE Trans. Inf. Theory}, vol.~56, no.~12, pp.
  6040--6057, Nov. 2010.

\bibitem{jafar2008degrees}
S.~A. Jafar and S.~Shamai, ``{Degrees of freedom region of the MIMO $ X $
  channel},'' \emph{IEEE Trans. Inf. Theory}, vol.~54, no.~1, pp. 151--170,
  Jan. 2008.

\bibitem{suh2008interference}
C.~Suh and D.~Tse, ``Interference alignment for cellular networks,'' in
  \emph{Proc. 46th Annu. Allerton Conf.}, Monticello, IL, Sep. 2008, pp.
  1037--1044.

\bibitem{motahari2014real}
A.~S. Motahari, S.~Oveis-Gharan, M.-A. Maddah-Ali, and A.~K. Khandani, ``{Real
  interference alignment: Exploiting the potential of single antenna
  systems},'' \emph{IEEE Trans. Inf. Theory}, vol.~60, no.~8, pp. 4799--4810,
  Aug. 2014.

\bibitem{jung2011opportunistic}
B.~C. Jung and W.-Y. Shin, ``{Opportunistic interference alignment for
  interference-limited cellular TDD uplink},'' \emph{IEEE Commun. Lett},
  vol.~15, no.~2, pp. 148--150, Feb. 2011.

\bibitem{jung2012opportunistic}
B.~C. Jung, D.~Park, and W.-Y. Shin, ``Opportunistic interference mitigation
  achieves optimal degrees-of-freedom in wireless multi-cell uplink networks,''
  \emph{IEEE Trans. Commun.}, vol.~60, no.~7, pp. 1935--1944, Jul. 2012.

\bibitem{yang2013opportunistic}
H.~J. Yang, W.-Y. Shin, B.~C. Jung, and A.~Paulraj, ``{Opportunistic
  interference alignment for MIMO interfering multiple-access channels},''
  \emph{IEEE Trans. Wireless Commun.}, vol.~12, no.~5, pp. 2180--2192, 2013.

\bibitem{yang2017opportunistic}
H.~J. Yang, W.-Y. Shin, B.~C. Jung, C.~Suh, and A.~Paulraj, ``Opportunistic
  downlink interference alignment for multi-cell {MIMO} networks,'' \emph{IEEE
  Trans. Wireless Commun.}, vol.~16, no.~3, pp. 1533--1548, Mar. 2017.

\bibitem{gou2012aligned}
T.~Gou, S.~A. Jafar, C.~Wang, S.-W. Jeon, and S.-Y. Chung, ``Aligned
  interference neutralization and the degrees of freedom of the $2 \times 2
  \times 2$ interference channel,'' \emph{IEEE Trans. Inf. Theory}, vol.~58,
  no.~7, pp. 4381--4395, Jul. 2012.

\bibitem{shomorony2014degrees}
I.~Shomorony and A.~S. Avestimehr, ``{Degrees of freedom of two-hop wireless
  networks: Everyone gets the entire cake},'' \emph{IEEE Trans. Inf. Theory},
  vol.~60, no.~5, pp. 2417--2431, May 2014.

\bibitem{zanella2017relay}
A.~Zanella, A.~Bazzi, and B.~M. Masini, ``Relay selection analysis for an
  opportunistic two-hop multi-user system in a {Poisson} field of nodes,''
  \emph{IEEE Trans. Wireless Commun.}, vol.~16, no.~2, pp. 1281--1293, Feb.
  2017.

\bibitem{knopp1995information}
R.~Knopp and P.~A. Humblet, ``Information capacity and power control in
  single-cell multiuser communications,'' in \emph{Proc. IEEE Int. Conf.
  Commun. (ICC)}, Seattle, WA, Jun. 1995, pp. 331--335.

\bibitem{viswanath2002opportunistic}
P.~Viswanath, D.~N.~C. Tse, and R.~Laroia, ``Opportunistic beamforming using
  dumb antennas,'' \emph{IEEE Trans. Inf. Theory}, vol.~48, no.~6, pp.
  1277--1294, Aug. 2002.

\bibitem{sharif2005capacity}
M.~Sharif and B.~Hassibi, ``{On the capacity of MIMO broadcast channels with
  partial side information},'' \emph{IEEE Trans. Inf. Theory}, vol.~51, no.~2,
  pp. 506--522, Jan. 2005.

\bibitem{shin2012network}
W.-Y. Shin and B.~C. Jung, ``Network coordinated opportunistic beamforming in
  downlink cellular networks,'' \emph{IEICE Trans. Commun.}, vol.~95, no.~4,
  pp. 1393--1396, Apr. 2012.

\bibitem{nguyen2013multi}
H.~D. Nguyen, R.~Zhang, and H.~T. Hui, ``Multi-cell random beamforming:
  Achievable rate and degrees of freedom region,'' \emph{IEEE Trans. Sig.
  Process.}, vol.~61, no.~14, pp. 3532--3544, Jul. 2013.

\bibitem{qin2006distributed}
X.~Qin and R.~A. Berry, ``Distributed approaches for exploiting multiuser
  diversity in wireless networks,'' \emph{IEEE Trans. Inf. Theory}, vol.~52,
  no.~2, pp. 392--413, Jan. 2006.

\bibitem{adireddy2005exploiting}
S.~Adireddy and L.~Tong, ``Exploiting decentralized channel state information
  for random access,'' \emph{IEEE Trans. Inf. Theory}, vol.~51, no.~2, pp.
  537--561, Jan. 2005.

\bibitem{lin2017multi}
H.~Lin and W.-Y. Shin, ``Multi-cell aware opportunistic random access,'' in
  \emph{Proc. IEEE Int. Symp. Inf. Theory (ISIT)}, Aachen, Germany, Jun. 2017,
  pp. 2538--2542.

\bibitem{lin2017MAORA}
------, ``Multi-cell-aware opportunistic random access for machine-type
  communications,'' \emph{IEEE Trans. Mobile Comput.}, submitted for
  publication, available at https://arxiv.org/abs/1708.02861.

\bibitem{cui2009opportunistic}
S.~Cui, A.~M. Haimovich, O.~Somekh, and H.~V. Poor, ``Opportunistic relaying in
  wireless networks,'' \emph{IEEE Trans. Inf. Theory}, vol.~55, no.~11, pp.
  5121--5137, Nov. 2009.

\bibitem{lin2016cognitive}
S.-C. Lin and K.-C. Chen, ``{Cognitive and opportunistic relay for QoS
  guarantees in machine-to-machine communications},'' \emph{IEEE Trans. Mobile
  Comput.}, vol.~15, no.~3, pp. 599--609, Mar. 2016.

\bibitem{shin2013parallel}
W.-Y. Shin, S.-Y. Chung, and Y.~H. Lee, ``Parallel opportunistic routing in
  wireless networks,'' \emph{IEEE Trans. Inf. Theory}, vol.~59, no.~10, pp.
  6290--6300, Oct. 2013.

\bibitem{gao2015forwarding}
W.~Gao, Q.~Li, and G.~Cao, ``{Forwarding redundancy in opportunistic mobile
  networks: Investigation, elimination and exploitation},'' \emph{IEEE Trans.
  Mobile Comput.}, vol.~14, no.~4, pp. 714--727, Apr. 2015.

\bibitem{so2017load}
J.~So and H.~Byun, ``Load-balanced opportunistic routing for duty-cycled
  wireless sensor networks,'' \emph{IEEE Trans. Mobile Comput.}, vol.~16,
  no.~7, pp. 1940--1955, Jul. 2017.

\bibitem{shin2017opportunistic}
W.-Y. Shin, V.~V. Mai, B.~C. Jung, and H.~J. Yang, ``Opportunistic network
  decoupling with virtual full-duplex operation in multi-source interfering
  relay networks,'' \emph{IEEE Trans. Mob. Comput.}, vol.~16, no.~8, pp.
  2321--2333, Aug. 2017.

\bibitem{bharadia2013full}
D.~Bharadia, E.~McMilin, and S.~Katti, ``Full duplex radios,'' \emph{ACM
  SIGCOMM Comput. Commun. Rev.}, vol.~43, no.~4, pp. 375--386, Oct. 2013.

\bibitem{day2012full}
B.~P. Day, A.~R. Margetts, D.~W. Bliss, and P.~Schniter, ``{Full-duplex MIMO
  relaying: Achievable rates under limited dynamic range},'' \emph{IEEE J. Sel.
  Areas Commun.}, vol.~30, no.~8, pp. 1541--1553, Sep. 2012.

\bibitem{kim2013distributed}
T.~M. Kim, H.~J. Yang, and A.~J. Paulraj, ``{Distributed sum-rate optimization
  for full-duplex MIMO system under limited dynamic range},'' \emph{IEEE Signal
  Process. Lett.}, vol.~20, no.~6, pp. 555--558, Jun. 2013.

\bibitem{knuth1976big}
D.~E. Knuth, ``{Big Omicron and big Omega and big Theta},'' \emph{ACM SIGACT
  News}, vol.~8, no.~2, pp. 18--24, Apr.--Jun. 1976.

\bibitem{hassibi2002multiple}
B.~Hassibi and T.~L. Marzetta, ``{Multiple-antennas and isotropically random
  unitary inputs: The received signal density in closed form},'' \emph{IEEE
  Trans. Inf. Theory}, vol.~48, no.~6, pp. 1473--1484, 2002.

\bibitem{bletsas2006simple}
A.~Bletsas, A.~Khisti, D.~P. Reed, and A.~Lippman, ``A simple cooperative
  diversity method based on network path selection,'' \emph{IEEE J. Selec.
  Area. Commun.}, vol.~24, no.~3, pp. 659--672, Mar. 2006.

\bibitem{gradshteyn2006table}
I.~S. Gradshteyn and I.~M. Ryzhik, \emph{Table of integrals, series, and
  products}, 6th~ed.\hskip 1em plus 0.5em minus 0.4em\relax San Diego, CA:
  Academic, 2006.

\bibitem{lin2017apwcs}
H.~Lin, W.-Y. Shin, and B.~C. Jung, ``Achieving optimal degrees of freedom in
  multi-source interfering relay networks with multiple antennas,'' in
  \emph{Proc. IEEE Veh. Technol. Soc. Asia Pacific Wireless Commun. Symp.},
  Incheon, Korea, Aug. 2017, pp. 1--2.

\end{thebibliography}
\end{document}